\documentclass[a4paper,10pt]{article}
\usepackage{cite,amsmath,amsfonts,amsthm,fullpage}
\usepackage{youngtab}
\newcommand{\Pf}{\mathop\mathrm{Pf}\nolimits}
\newcommand{\sgn}{\mathop\mathrm{sgn}\nolimits}
\newcommand{\Pa}{\mathop\mathrm{P}\nolimits}

\newcommand{\bpow}{\mathbf{p}}
\newcommand{\bbpow}{{\bar{\mathbf{p}}}}
\newcommand{\bt}{\mathbf{t}}

\theoremstyle{plain}

\newtheorem{Theorem}{Theorem}
\newtheorem{Lemma}{Lemma}
\newtheorem{Proposition}{Proposition}
\newtheorem{Corollary}{Corollary}
\newtheorem{Remark}{Remark}

\theoremstyle{remark}

\def\l{\langle}
\def\r{\rangle}

\def\tr{\mathrm {tr}}
\def\det{\mathrm {det}}

\def\diag{\mathrm {diag}}
\def\res{\mathop{\mathrm {res}}\limits_}

\def\bp{\begin{Proposition}}
\def\ep{\end{Proposition}}
\def\bc{\begin{Corollary}}
\def\ec{\end{Corollary}}
\def\bl{\begin{Lemma}}
\def\el{\end{Lemma}}
\def\be{\begin{equation}}
\def\ee{\end{equation}}
\def\br{\begin{Remark}\rm\small}
\def\er{\end{Remark}}
\def\brs{\begin{remarks}.\\ \rm\
\begin{enumerate}}
\def\ers{\end{enumerate}\end{remarks}}
\def\bea{\begin{eqnarray}}
\def\eea{\end{eqnarray}}

\def\tr{\mathrm {tr}}
\def\det{\mathrm {det}}

\def\sgn{\mathrm {sgn}}

\def\diag{\mathrm {diag}}

\def\res{\mathop{\mathrm {res}}\limits}

\def\&{&{\hskip -20pt}}

\newcount\YDcount\YDcount=0
\def\YDsize{10pt}

\def\YD#1{%
\ifnum#1=0
 \ifnum\YDcount=0 \ifx\varnothing\undefined\emptyset\else\varnothing\fi
 \else\vskip1.4pt\egroup\YDcount=0\fi
\else
 \ifnum\YDcount=0 \YDcount=1\vcenter\bgroup\vskip1pt
 \else\nointerlineskip\fi
 \vbox{\hrule\hbox{\vrule height\YDsize
 \loop\hskip\YDsize\vrule\ifnum\YDcount<#1\advance\YDcount1\repeat}\hrule
 \kern-0.4pt}\expandafter\YD
\fi}

\begin{document}
\author{ S.M.  Natanzon\thanks{National Research University Higher School of Economics, Moscow, Russia; Institute for 
Theoretical and Experimental Physics, Moscow, Russia;
Laboratory of Quantum Topology, Chelyabinsk State University, Chelyabinsk, Russia; email: natanzons@mail.ru} \and A. Yu.
Orlov\thanks{Institute of Oceanology, Nahimovskii Prospekt 36,
Moscow 117997, Russia, and National Research University Higher School of Economics,
International Laboratory of Representation
Theory and Mathematical Physics,
20 Myasnitskaya Ulitsa, Moscow 101000, Russia, email: orlovs@ocean.ru
}}
\title{Hurwitz numbers and BKP hierarchy}

\maketitle

\begin{abstract}

We consider special series in ratios of the Schur functions which are defined by integers $\textsc{f}\ge 0$ and 
$\textsc{e} \le 2$,
and also by the set of $3k$ parameters $n_i,q_i,t_i,\,i=1,\dots, k$. These series may be presented in form of matrix integrals.
In case $k=0$ these series generates
Hurwitz numbers for the $d$-fold branched covering of connected surfaces with a given Euler characteristic $\textsc{e}$ 
 and arbitrary profiles at $\textsc{f}$ ramification points. If $k>0$ they generate weighted sums of the Hurwitz numbers with 
 additional ramification points which are 
distributed between color groups indexed by $i=1,\dots,k$, the weights being written in terms of parameters $n_i,q_i,t_i$. 
By specifying the parameters we get sums of all Hurwitz numbers with $\textsc{f}$ arbitrary fixed profiles and the additional 
profiles provided the following condition: both, the sum of profile lengths and the number of ramification 
points in each color group are given numbers.  In case $\textsc{e}=\textsc{f}=1,2$ the series may be identified with 
 BKP  tau functions of Kac and van de Leur of a special type called hypergeometric tau functions. Sums of Hurwitz numbers 
 for $d$-fold branched coverings of ${\mathbb{RP}}^2$ are related to the one-component BKP hierarchy. We also present 
 links between sums of Hurwitz numbers and one-matrix model of the fat graphs.

\end{abstract}

\bigskip

\textbf{Key words:} Hurwitz numbers, tau functions, matrix integrals, BKP, multi-component KP, 
non-orientable surfaces, projective plane, Schur functions, 
hypergeometric functions, random partitions

\section{Introduction}

In the beautiful paper  \cite{Okounkov-2000}, A. Okounkov studied
ramified coverings of the Riemann sphere with arbitrary ramification
type over $0$ and $\infty$, and simple ramifications elsewhere, and it was proved that the
generating function for the related Hurwitz numbers (numbers of nonequivalent coverings with
given ramification type) is a tau -function for the Toda lattice hierarchy. In further works
\cite{GJ, AMMN-2014, Zog, Chekhov-2014, KZ} other examples of tau-functions for 2D Toda and KP hierarchy
generating Hurwitz numbers of the sphere were constructed. In recent work \cite{HO-2014} there was considered
more general TL tau function which gives some new examples of what were called composite signed Hurwitz
numbers and includes previous examples.

In the present paper we consider Hurwitz numbers of the projective plane $\mathbb{RP}^2$.
For our purpose we use the BKP hierarchy of integrable equations introduced by V.Kac and J. van de Leur in
\cite{KvdLbispec}.
We also present the most general weighted combinations of Hurwitz numbers for the sphere $\mathbb{CP}^1$ which
may be related to the two-component KP hierarchy which generalizes results of \cite{HO-2014}.
(A brief explanation what happens when we change the hierarchy is given in the next paragraph.)
Our main result is that the BKP tau function (\ref{BKP-generating-Hurwitz-sums}) is the generating function for
the certain linear combinations of Hurwitz numbers of the projective plane, 
see Theorem \ref{Th-BKP} and relations  (\ref{T-BKP}), (\ref{special'''''}). 

\vspace{2ex}

For readers familiar with the topic, let us briefly explain the difference between TL and BKP tau functions
in the context of generating of Hurwitz numbers and topology of the base.

The Frobenius formula for the Hurwitz numbers enumerating $d$-fold branched coverings of Riemann or Klein
surfaces contains the sum over irreducible representations $\lambda$ of the symmetric group
(see  \cite{F,FS,M1,M2,ZL,GARETH.A.JONES})
\be\label{Hurwitz-counting}
H_\Omega(\Delta^{(1)}\,\dots,\Delta^{(\textsc{f})})\,=\,d!
\sum_{\chi} \,\,\left(
\prod_{i=1}^\textsc{f} \,\,|C_{\Delta^{(i)}}|\,\,\frac{\chi(\Delta^{(i)})}{\chi(1)} \right)\,\,
\left(\frac{\chi(1)}{d!}\right)^\textsc{e},
\ee
where $\textsc{e}$ is the Euler characteristic of $\Omega$, $\{\Delta^{(i)}\}$,
$\Delta^{(i)}$ are profiles  over ramification points on $\Omega$, $\chi(\Delta)$ is a character of the symmetric group $S_d$ evaluated at a cycle type $\Delta$, and $\chi$ ranges over the irreducible complex characters of $S_d$. Each profile $\Delta^{(i)}$ is a partition of $d$ - the set of non-negative non-increasing numbers
$(d^{(i)}_1,d^{(i)}_2,\dots )$, which describes the ramification over the point number $i$ on the base. The weights
of all partitions involved in (\ref{Hurwitz-counting}) are equal: $|\Delta^{(i)}|=d$. The number $|C_{\Delta}|$ is the number of elements in the cycle class $\Delta$ in $S_d$.

The formula (\ref{Hurwitz-counting}) was derived for connected $\Omega$ which in particular implies that
$\textsc{e}\le 2$.
The geometrical meaning of the formula (\ref{Hurwitz-counting}) in  case $\textsc{e}>2$  is unclear.
We shall use the notation $H^\textsc{e}$ instead of $H_\Omega$ to denote the left hand side of  (\ref{Hurwitz-counting})
where we allow $\textsc{e}$ to be any integer.

The Hurwitz numbers form a topological field theory \cite{Dijkgraaf}. They are used in mathematical physics
(for instance in \cite{Dijkgraaf}) and in algebraic geometry \cite{ZL}.

Here and below we write $\chi_\lambda(1)$ having in mind the evaluation of the irreducible character
of the symmetric group $\chi_\lambda$ at the unity element in the symmetric group, which is given by the
partition $(1^{d})$, then, $\chi_\lambda(1)={\rm dim} \lambda$.

It is well-known that Schur functions $s_\lambda$ and characters of the symmetric group $\chi_\lambda$
are linearly dependent \cite{Mac}.
Soliton theory provides various series of products of the Schur
functions over partitions for tau functions of various hierarchies of integrable equations. In \cite{Okounkov-2000} Okounkov introduced and studied the following sum
\be\label{Okounkov-tau}
 \sum_\lambda e^{\beta f(\lambda,n)}
 \,s_\lambda(\bpow)s_\lambda(\bbpow)
\ee
 with
 \be
 \label{Okounkov-f}
 f(\lambda,n)=n|\lambda| +  |C_{\Gamma}|\,\frac{\chi_\lambda(\Gamma)}{\chi_\lambda(1)}
\ee
where both, the partition $\Gamma=\Gamma_d:=(1^{d-2}2)$ and the number $|C_{\Gamma}|$ depends only on $d$, 
see (\ref{cardinality}) below. He shown that series (\ref{Okounkov-tau}) is a
tau function of the Toda lattice where power sum variables $\bpow=(p_1,p_2,\dots)$ and
$\bbpow=({\bar p}_1,{\bar p}_1,\dots)$ together with the integer $n$ play the role of higher times.

Moreover, (\ref{Okounkov-tau}) generates a certain class of the Hurwitz numbers (\ref{Hurwitz-counting}).
This class
describes coverings of the Riemann sphere ($\textsc{e}=2$) with arbitrary given profiles $\bpow$, $\bbpow$ over two given
points (say, $0$ and $\infty$) and any number of simple ramifications described by the Young diagram $\Gamma$
(here and below we shall omit the dependence of $\Gamma$ on $d$).
One can recognize it thanks to the the relation between the Schur functions and the characters of the symmetric
group \cite{Mac} which we re-write in a suitable form as follows:
\be\label{char-map-introduction}
s_\lambda(\bpow)= \, p_1^d
\left(1+
\, p_1^{-d}\sum_{\Delta\neq 1^d}\,\,|C_{\Delta}|\,\,
\frac{\chi_\lambda(\Delta)}{\chi_\lambda(1)}\,\,\bpow_\Delta \,
\right)\, \frac{\chi_\lambda(1)}{d!}
\ee
where the summation ranges over all partitions $\Delta=(d_1,d_2, \dots)$  of the number $d=|\lambda|$, and where
$\bpow_\Delta$  is the product $p_{d_1}p_{d_2}\cdots $. The numbers $|C_\Delta|$ depend only on $\Delta$ (see \cite{Mac} or
(\ref{cardinality}) below).

Now, it is clear from (\ref{char-map-introduction}) that the formula (\ref{Okounkov-tau}) is a generating function
for Hurwitz numbers (\ref{Hurwitz-counting}) where
 $\beta$, $\bpow$ and $\bbpow$ are formal parameters. Basically,
the Taylor coefficients in the terms $p_\Delta{\bar p}_{\Delta'}\beta^b$, up to a factor, coincide with the
number of covers with the ramification type $\Delta,\Delta'$ over two points and further of type
$\Gamma=1^{|\Delta|-2}2$ over $b$ points.

\vspace{2ex}

We present a larger class of series of type (\ref{Okounkov-tau}), which we shall call generating Hurwitz series:
\be\label{hyp-tau-e-f}
\tau^{(\textsc{e},\textsc{f})}\left(N,n,\bpow^{(1)},\dots,\bpow^{(\textsc{f})}|\{q_i,t_i,n_i \},\beta\right)=
\sum_{\lambda \atop \ell(\lambda)\le N} \,\,r_\lambda^{{\bf q},{\bf t},{\bf n},\beta}(n)\,
\left(\chi_\lambda(1) \right)^{\textsc{e}-\textsc{f}}
\prod_{j=1}^\textsc{f} s_\lambda\left(\bpow^{(j)}\right)\,,\quad \textsc{e} \in \mathbb{Z}
\ee
where $\bpow^{(j)}=(p_1^{(j)},p_2^{(j)},\dots)$ are power sums variables, $a_i,q_i,n_i\in\mathbb{C}$, and
\be\label{r}
r_\lambda^{{\bf q},{\bf t},{\bf n},\beta}(n)=\,e^{\beta f(\lambda,n)}\,
\prod_{i=1}^k \left(\frac{s_\lambda(\bpow(q_it_i^n,t_i))}{s_\lambda(\bpow(0,t_i))}\right)^{n_i}
\ee
while power sum variables $\bpow(q_i,t_i)=\left(p_1(q_i,t_i),p_2(q_i,t_i),\dots  \right)$ are specified as follows
\be\label{p-m-a-i-q-i}
p_m(q_i,t_i):=\frac{1-q_i^{ m}}{1-t_i^m},\quad
p_m(0,t_i):=\frac{1}{1-t_i^m}
\ee

On the relation of these $\bpow(q_i,t_i)$ to Macdonald polynomials see the Appendix \ref{MacdPol}.
We have
\be\label{ratio-of-s-s=ratio-of-chi}
\left(\frac{s_\lambda(\bpow(q_i,t_i))}{s_\lambda(\bpow(0,t_i))}\right)^{n_i}=\left(1-q_i\right)^{n_i d}
\left(\frac{1+
\, \sum_{\Delta\neq 1^d}\,\,|C_{\Delta}|\,\,
\frac{\chi_\lambda(\Delta)}{\chi_\lambda(1)}\,w_\Delta(q_i,t_i) \,}{1+
\, \sum_{\Delta\neq 1^d}\,\,|C_{\Delta}|\,\,
\frac{\chi_\lambda(\Delta)}{\chi_\lambda(1)}\,w_\Delta(0,t_i)}\right)^{n_i}
\ee
where
\be\label{w-introduction}
w_\Delta(q_i,t_i)=\left(\frac{1-t_i}{1-q_i}\right)^{d}
\prod_{j=1}^{\ell} \frac{1-q_i^{d_j}}{1-t_i^{d_j}}\,,\qquad
\Delta=(d_1,\dots,d_\ell)
\ee

For the sake of brevity we shall also write $\tau_r^{(\textsc{e},\textsc{f})}$ instead of
$\tau^{(\textsc{e},\textsc{f})}\left(N,n,\bpow^{(1)},\dots,\bpow^{(\textsc{f})}|\{q_i,t_i,n_i \},\beta\right)$.

The series (\ref{hyp-tau-e-f}) are interesting because of three aspects.

First, for $\textsc{f}=\textsc{e}=1,2$ they may be related to the $\textsc{e}$-component BKP hierarchy of integrable equations.
For $\textsc{e}=2$, they also may be related to the 2-component KP hierarchy (2KP) and Toda lattice hierarchy. 
The cases $\textsc{e}=1$
(one-component BKP) and $\textsc{e}=2$
(2-component KP, and also 2-component BKP) describes specializations
of tau functions of hypergeometric type studied respectively in \cite{OST-I} and in \cite{KMMM}, \cite{OS-2000}
(in case $n_i=\pm 1$, such specialization was introduced in \cite{OS-2000}, \cite{OS-TMP} in the context of hypergeometric
functions of matrix argument \cite{GRich} and Milne's hypergeometric functions \cite{Milne}.

Second, by (\ref{char-map-introduction}) series (\ref{hyp-tau-e-f}) may be viewed as a generating function
for certain weighted combinations of the right hand sides of eq. (\ref{Hurwitz-counting}), $H^\textsc{e}$.
To see it, again, we use (\ref{char-map-introduction}) replacing each Schur function by its character
expansion. One can see that thanks to (\ref{char-map-introduction}) we obtain linear combinations
of terms $H^\textsc{e}$. We notice that the factor $r_\lambda$
in (\ref{hyp-tau-e-f}) contains only ratios of the Schur functions and due to (\ref{char-map-introduction})
does not contribute to the power of $\chi_\lambda(1)$ in (\ref{Hurwitz-counting}), this power is the result
of the multiplication of the $\textsc{e}$ Schur functions in the left hand side of (\ref{hyp-tau-e-f}).

At last, series (\ref{hyp-tau-e-f}) may be obtained as integrals of 2KP and BKP tau functions over matrices.

Let us note that series (\ref{hyp-tau-e-f}) where $\textsc{e}=2$ and $n_i=1,\,q_i=t^{a_i}_i\to 1,\,i=1,\dots,k$
we introduces in \cite{AMMN-2014}, and the case $\textsc{e}=2$ and $n_i=\pm,\,q_i=t^{a_i}_i\to 1,\,i=1,\dots,k$
was studied in \cite{HO-2014}.

The generating Hurwitz series labeled by $\textsc{e}-1,\textsc{f}-1$ may be obtained from the series labeled by
$\textsc{e},\textsc{f}$ as follows
\be
\tau^{(\textsc{e}-1,\textsc{f}-1)}\left(N,n,\bpow^{(1)},\dots,\bpow^{(\textsc{e}-1)}\right)=
\left[ e^{L^{(\textsc{e})}}\cdot\tau^{(\textsc{e},\textsc{f})}\left(N,n,\bpow^{(1)},\dots,
\bpow^{(\textsc{e})}\right)\right]_{\bpow^{(\textsc{e})}=0}
\ee
where $L^{(\textsc{e})}$ is the Laplace operator
$\sum_{m>0} \left( \frac m2 \frac{\partial^2}{\partial p_m^2} + \frac{\partial}{\partial p_{2m-1}} \right)$,
where each $p_m$ is $p_m^{(\textsc{e})}$,
see Section \ref{transformations-section}.

Now for $\textsc{e}=\textsc{f}=2$ we have  the following series 
\be\label{general-hyp-tau-TL}
\tau^{2KP}(n,\bpow,\bbpow)=\sum_\lambda \,\,r_\lambda(n)\,
s_\lambda(\bpow)s_\lambda(\bbpow)
\ee
which may be identified with TL or, the same, 2KP tau function, see Section \ref{BKP-section} below.

 As a result one may conclude that, similar to the Okounkov tau function, the TL tau function (\ref{general-hyp-tau-TL}) generates
linear combinations of number of covering of the sphere. These specific combinations of Hurwitz numbers
(``composite signed Hurwitz numbers'')
in case $q_i\to 1$ are written down in \cite{HO-2014}
and for our convenience are also reproduced in the present text.

\vspace{2ex}

Put $\textsc{e}=\textsc{f}=1$. The series
\be
\label{general-hyp-tau-BKP}
\tau^{BKP}(n,\bpow)=\sum_\lambda \,\,r_\lambda(n)\,s_\lambda(\bpow)
\ee
with the same $r_\lambda(n)$ is also a tau function \cite{OST-I} where the set $\bpow$ plays the role of higher times ,
but now it is a different hierarchy, namely, the BKP hierarchy introduced by Kac and van de Leur \cite{KvdLbispec}.
This case will be considered in the present paper. It is easy to see that now $\textsc{e}=1$. Thus, according to
Frobenius formula (\ref{Hurwitz-counting}), function (\ref{general-hyp-tau-BKP}) is a generating function for Hurwitz
numbers of the projective plane $\mathbb{RP}^2$.

If we consider the further case $\textsc{e}=0$ (the coverings of the elliptic curve) with the help of the series
\be
\label{trace}
\tau=\sum_{n,\lambda} \,\,r_\lambda(n)
\ee
then we see it is not a tau function as there are no time variables here. This expression may be related to the trace
of the certain diagonal ${\hat GL}_\infty$ element in the fermionic Fock space.

Thus, on the formal level we can explain the appearance of different hierarchies of integrable equations in the description
of Hurwitz counting problem.

\vspace{2ex}

This paper is the detalization of the consideration above.
We shall study only the case all $q_i\to 1$. The paper is organized as follows. In Section  \ref{Hurwitz-numbers-section} we
describe Hurwitz numbers and its combinations generating solutions of integrable systems. In Section \ref{BKP-section} we
recall some facts about BKP and TL hierarchies. We need a special class of tau functions which we call hypergeometric tau
functions.
 In Section \ref{Particular cases} in particular we review TL tau functions generating composite signed Hurwitz numbers
according to \cite{HO-2014}. However we need a modification caused by semiinfinity of TL which we need to compare results
with the BKP case later in Section \ref{transformations-section}. In sections \ref{Hurwitz-BKP-section} and
\ref{transformations-section} we construct Hurwitz $\tau$-functions for BKP hierarchy and find its connection with Hurwitz
$\tau$-functions for the semiinfinte 2DToda hierarchy.

In Section \ref{Matrix-integrals} we present a relation between fat graph counting obtained from the one-matrix
models and sums of Hurwitz numbers, and also write down an analogue of the one-matrix model which generates
similar sums for Hurwitz numbers of the projective plane.
We show ways to get Hurwitz generating series (\ref{hyp-tau-e-f})
in form of matrix integrals. These are integrals of the simplest 2KP and BKP functions whose as functions of products of matrices.
One of the way to diminish the Euler $\textsc{e}$ of the generating series by 1 is to replace one of 2KP tau functions
in the integrand by a BKP tau function.

In the Appendix \ref{Matrix-integrals-appendix} we discuss Hurwitz generating series in the context of matrix integrals.
In the Appendix \ref{fermionic-appendix} we write down the fermionic expression of the Hurwitz generating series (\ref{hyp-tau-e-f})
where $\textsc{e=f}>2$.

\section{Weighted sums of Hurwitz numbers \label{Hurwitz-numbers-section}}

\subsection{Hurwitz numbers}

For a partition $\Delta$ of a number $d=|\Delta|$ denote by $\ell(\Delta)$ the number of the non-vanishing parts.
For the Young diagram, corresponding to $\Delta$, the number $|\Delta|$ is the weight  of the diagram  and $\ell(\Delta)$
is the number of rows. Denote by $(d_1,\dots,d_{\ell})$ the Young diagram with rows of length $d_1,\dots,d_{\ell}$ and
corresponding partition of $\sum d_i$.

Hurwitz number $H_\Omega(d,\Delta^{(1)},\dots,\Delta^{(\textsc{f})})$ is defined by a connected surface $\Omega$ and partitions
$\Delta^{(1)},\dots,\Delta^{(\textsc{f})}$ of the number $d=|\Delta^{(i)}|,\,i=1,\dots,\textsc{f}$. The Hurwitz number
$H_\Omega(d,\Delta^{(1)},\dots,\Delta^{(\textsc{f})})$ is the
weighted number of branched coverings of the surface $\Omega$ by other surfaces (connected or non-connected) with
fixed critical values $z_1,\dots,z_\textsc{f}\in\Omega$ of topological types $\Delta^{(1)},\dots,\Delta^{(\textsc{f})}$.
More precisely,
$z\in\Omega$ is the critical value of the branched covering $f:\Sigma\rightarrow\Omega$ if $z=f(p)$, where $p\in\Sigma$ is
a critical point of $f$. Consider degrees $d_1,\dots,d_l$ of $f$ in all preimiges $f^{-1}(z)$. The partition  $(d_1,\dots,d_\ell)$
of $d=\deg(f)$ is called the topological type of the critical value $z$. We say that branched coverings
$f':\Sigma'\rightarrow\Omega$ and $f'':\Sigma''\rightarrow\Omega$ are the same, if there exists a homeomorphism
$g:f'\rightarrow f''$ such that $f'=f''g$. Then
\be H_\Omega(d,\Delta^{(1)},\dots,\Delta^{(\textsc{f})})=\sum\frac{1} {|\texttt{Aut}(f)|}\quad,
\ee
where the sum is taken over all branched coverings $f$ of $\Omega$, with the critical values $z_1,\dots,z_\textsc{f}\in\Omega$ of the
topological types $\Delta^{(1)},\dots,\Delta^{(\textsc{f})}$ respectively. This number is independent of the positions of the
branching points $z_i$.

The Hurwitz numbers arise in different fields of mathematics: from algebraic geometry to integrable systems. They are well studied
for orientable $\Omega$. In this case the Hurwitz number coincides with the weighted number of holomorphic branched coverings of a
Riemann surface $\Omega$ by another Riemann surfaces, having critical values $z_1,\dots,z_\textsc{f}\in\Omega$ of topological types
$\Delta^{(1)},\dots,\Delta^{(\textsc{f})}$ respectively. The well known isomorphism between Riemann surfaces and complex a
lgebraic curves gives the interpretation of the Hurwitz numbers as the numbers of morphisms of complex algebraic curves.

In this work we consider the Hurwitz numbers for non-orientable $\Omega$ without boundary. They have also two other
interpretations: as the numbers of the branched coverings of a Klein surface without boundary by another Klein surface, and as the
number of morphisms of real algebraic curves without real points. Klein surfaces are factors of Riemann surfaces by antiholomorphic
involutions. They correspond to real algebraic curves. Real points of real curves correspond to fixed points of the involutions
and boundary points of the Klein surfaces (see \cite{AG,N90,N2004}). In this paper we consider only surfaces without boundaries.
But an analog of the Hurwitz numbers for surfaces with boundaries also exists (\cite{AN,N}).

The Hurwitz numbers are closely connected with irreducible representations of $S_d$. The action of any permutation $s\in S_d$
split the set ${1,\dots,d}$ on subsets cardinality $(d_1,\dots,d_\ell)$ and thus generate a partition
$\Delta(s)=(d_1,\dots,d_\ell)$ of $d$. This partition is called as cyclic type of $s$. Conversely, any partition $\Delta$ of $d$ generate the set $C_{\Delta}\subset S_d$, consisted of permutation of cyclic type $\Delta$. The cardinality of $C_\Delta$ is equal to
\be\label{cardinality}
|C_\Delta| \,=\,\frac{|\Delta|!}{z_\Delta}\,,\qquad
z_\Delta\,=\,\prod_{i=1}^\infty \,i^{m_i}\,m_i!
\ee
where $m_i$ denotes the number of parts equal to $i$ of the partition $\Delta$ (then a partition $\Delta$ is often
denoted by $1^{m_1}2^{m_2}\cdots$).

Moreover, if $s_1,s_2\in C_{\Delta}$, then $\chi(s_1)=\chi(s_2)$ for any complex characters $\chi$ of $S_d$. Thus we can define $\chi(\Delta)$ for a partition $\Delta$, as $\chi(\Delta)=\chi(s)$ for $s\in C_\Delta$.

The Frobenius formula \cite{F,FS,M1,M2,ZL,GARETH.A.JONES} says that
\be \label{Hurwitz}
H_\Omega(d,\Delta^{(1)},\dots,\Delta^{(\textsc{f})})= \,d!
\sum_{\chi} \,\,\left(\prod_{i=1}^\textsc{f} \,\,|C_{\Delta^{(i)}}|\,\,\frac{\chi(\Delta^{(i)})}
{\chi(1)} \right)\,\,
\left(\frac{\chi(1)}{d!}\right)^{\textsc{e}},
\ee
where $\textsc{e}$ is the Euler characteristic of $\Omega$ and $\chi$ ranges over the irreducible complex characters of $S_d$,
associated with Young diagrams of wight $d$.

\vspace{2ex}

In what follows we shall construct the generating series for the numbers
\be\label{generating}
H^{\textsc{e},\textsc{f}}_{d,N}\left(\Delta^{(1)}\,\dots, \Delta^{(\textsc{f})} \right)\,:=\,d!
\sum_{\lambda \atop |\lambda|=d \,,\, \ell(\lambda)\le N} \,\,\left(
\prod_{i=1}^\textsc{f} \,\,|C_{\Delta^{(i)}}|\,\,\frac{\chi_\lambda(\Delta^{(i)})}{\chi_\lambda(1)} \right)\,\,
\left(\frac{\chi_\lambda(1)}{d!}\right)^{\textsc{e}},
\ee
depending on free integer parameters $\textsc{e},\textsc{f},d,N$. For $N \ge d$ the number $H^{\textsc{e},\textsc{e},}_{d,N}$
does not depend on $N$. In particular for $N\ge d$ we have
$H^{2,\textsc{f}}_{d,N}\left(\Delta^{(1)}\,\dots,\Delta^{(\textsc{f})}\right)
=H_{\mathbb{CP}^1}\left(d,\Delta^{(1)}\,\dots,\Delta^{(\textsc{f})}\right)$ and $
H^{1,\textsc{f}}_{d,N}\left(\Delta^{(1)}\,\dots,\Delta^{(\textsc{f})}\right)
=H_{\mathbb{RP}^2}\left(d,\Delta^{(1)}\,\dots,\Delta^{(\textsc{f})}\right)$.

\paragraph{Suitable notations: occupation numbers.}

Given $d$ we have a finite number, say denoted by $d^*+1$,
of all different partitions of $d$, $\Delta_0,\dots, \Delta_{d^*}$ (the value of $d^*=d^*(d)$ is unimportant for us).
For a given $d$ it is convenient to enumerate partitions by a label, say, $j=0,1,\dots,d^*$ where we reserve $j=0$ for
the partition $1^d$: $\Delta_0=1^d$. By $c_j \ge 0$ denote how many times a
partition $\Delta_j$ occurs in the set of arguments of $H^\textsc{e}_{d,N}$ and write
\[
H_\Omega ( d;
\underbrace{\Delta_0,\dots,\Delta_0 }_{c_0},\dots,
\underbrace{
\Delta_{d^*},\dots,\Delta_{d^*}
}_{c_{d^*}} )\,=:\,
H_\Omega (d;{\bf c})
\]
where ${\bf c}={\bf c}(d)=(c_0,\dots,c_{d^*-1})$.
By analogue with particle and statistical physics we will call $c_j$ the occupation number of the state $j$
(similar notations were used in \cite{MM2}).
As one concludes from (\ref{Hurwitz}) the Hurwitz numbers do not depend on $c_0$, and it is obvious from the
geometrical point of view since it is related to the absence of branching.

We shall also use the mixed notations, where we split  the arguments of $H_\Omega$ into two groups:
\[
H_\Omega\left(d;{\bf c},\Delta^{(1)},\dots, \Delta^{(\textsc{f})}\right):=\,
H_\Omega\left(d;\underbrace{\Delta_0,\dots,\Delta_0}_{c_0},\dots,
\underbrace{\Delta_{d^*},\dots,\Delta_{d^*}}_{c_{d^*}},\Delta^{(1)},\dots, \Delta^{(\textsc{f})}\right)
\]

We shall consider weighted sums of Hurwitz numbers over the occupation numbers keeping a group of
$\textsc{f}$ partitions fixed. The series (\ref{hyp-tau-e-f}) yields examples of such sums,
which have the following form 
\[
\tau^{(\textsc{e},\textsc{f})}\left(N,n,\bpow^{(1)},\dots,\bpow^{(\textsc{f})}|\{q_i,t_i,n_i \},\beta\right)
:= \sum_{d}\,\sum_{\bf c}\,\omega\left(d,{\bf c}\right)\,
H^{\textsc{e},\textsc{f}+|{\bf c}|}_{d,N}({\bf c},\Delta^{(1)},\dots, \Delta^{(\textsc{f})})
\prod_{i=1}^{\textsc{f}}\bpow^{(i)}_{\Delta^{(i)}}
\]
where the weights $\omega$ depend on $n,\bpow^{(j)},\beta$ and on $q_i,t_i,n_i,\,i=1,\dots,k$, 
and will be found in the next section.

\subsection{Weighted sums of Hurwitz numbers. \label{Weighted-sums-Section}}

Our goal is to explain what kind of information related to Hurwitz number 
is hidden in series (\ref{hyp-tau-e-f}) depending on $3k$ parameters $q_i,t_i,n_i$.

\paragraph{The case $k=0$, $b=0$.} First of all we notice that if the factor $r_\lambda$ is equal to 1, the series
\[
 \tau^{\textsc{e,f}}_1(N,\bpow^{(1)},\dots, \bpow^{(\textsc{f})})\,=\,\sum_{\lambda\in\Pa\atop \ell(\lambda)\le N}\,
 \left( s_\lambda(\bpow_\infty)\right)^\textsc{e}
 \, \prod_{i=1}^\textsc{f}\,\frac{s_\lambda(\bpow^{(i)}))}{s_\lambda(\bpow_\infty)}
\]
generate Hurwitz numbers themselves. Here we use the notation $\bpow_\infty=(1,0,0,\dots)$, then, from 
(\ref{char-map-introduction}) we obtain the known relation \cite{Mac}
\be\label{dim-lambda-Schur}
  s_\lambda(\bpow_\infty)\, =\,\frac{{\rm dim}\,\lambda}{|\lambda|!}\,=\, \frac{\chi^\lambda_{(1^{|\lambda|})}}{|\lambda|!}
\ee
where ${\rm dim} \lambda$ is the dimension of the irreducible representation $\lambda$ of the symmetric group 
$S_d$, $d=|\lambda|$.

Indeed, thanks to (\ref{char-map-introduction}) we easily obtain
\be
 \tau^{\textsc{e},\textsc{f}}_1(\bpow^{(1)},\dots, \bpow^{(\textsc{f})})\,=\,
 \sum_{\Delta^{(1)},\dots,\Delta^{(\textsc{f})}}\,
 H^{\textsc{e},\textsc{f}}_{d,N}\left(d;\Delta^{(1)}\,\dots,\Delta^{(\textsc{f})}\right) \,
 \prod_{i=1}^\textsc{f}\,\bpow^{(i)}_{\Delta^{(i)}}
\ee
where $H^{\textsc{e},\textsc{f}}_{d,N}$ is given by (\ref{generating}).

\paragraph{The case $k,b>0$.} In case the prefactor $r_\lambda$ is not identical to $1$ the series (\ref{hyp-tau-e-f})
generates not the Hurwitz numbers but certain linear combinations (weighted sums) of these numbers.
This is the subject of Proposition \ref{tau=Hurwitz-sums} below. The contribution of the exponential prefactor
in (\ref{r}) is already explained in \cite{Okounkov-2000} and we concentrate on the contribution of the ratio of the Schur 
functions in (\ref{r}).

Let us note that we need some preliminary work to explain what sort of weighted sums we are going to obtain.
These are weighted sums over additional partitions. Each additional partition belongs to one of $2k$ colored group.
The number of the partitions in each group is not fixed, the weight is defined by the partition and by the color $i$,
where $i=1,\dots,2k$, of the partition, namely, the weight depends on the values of the complex parameters $q_i,t_i,n_i$.
 
We need additional notations for the additional partitions. \footnote{The notations $\Gamma=1^{d-2}2$ and 
$\Delta^{(i)},\,i=1,\dots,\textsc{f}$ we will keep for fixed partitions, while for additional ones we shall 
use subscripts and superscripts without brackets like $\Delta^i_j$ below.}
First of all we notice that the additional partitions are profiles of the branch points and all profiles
has the same profile weight (the weight of the related partition) equal to $d$.
Denote the set of all partitions of the weight $d$
by $\Pa_d$. We recall: $\{ \Delta_j\in\Pa_d\,, \,j=0,\dots,d^*\}$ is the complete set of different partitions of the 
weight $d$, and $\Delta_0=1^d$.

Let us consider a set $D\in\Pa^{\times 2k}_d$  of partitions of the same
weight $d$ separated into $k$ pairs of (color) groups. Color groups are numbered by $k$ pairs of indexes 
$1,{\bar 1},\dots,k,{\bar k}$.
Then each partition from this set may be indexed by a superscript 
which numbers the color of he partition and by a subscript which indexes its number among all partitions of a given weight $d$. 
Say $\Delta^i_j$ ($\Delta^{\bar i}_j$) is the partition numbered by $j$ in the set of all partitions of $d$ from the color 
group $i$ ($\bar i$). 
Then the occupation numbers
$c^i_j$ (and ${\bar c}^i_j$) where $i=1,\dots,k$, $j=0,\dots,d^*-1$ says how many times the profile $\Delta_j$ occurs 
in the color group $i$ (respectively ${\bar i}$).
Thus $D=D\left({\bf c}^1,{\bf {\bar c}}^1,\dots,{\bf c}^k,{\bf {\bar c}}^k\right)$ is defined by the $2k$ sets of occupation
numbers: 
by ${\bf c, {\bar c}}^i=(c^i_0,{\bar c}^i_0,\dots,c^i_{d^*-1},{\bar c}^i_{d^*-1})$, 
$i=1,\dots,k$.

We introduce the notation $|{\bf c}^i|=\sum_{j=1}^{d^*-1}\,c^i_j$ (and  $|{\bf {\bar c}}^i|=\sum_{j=1}^{d^*-1}\,c^i_j$)
which says how many partitions different from $\Delta_0$ is contained in the color group $i$ (resp. ${\bar i}$) of
$D\left({\bf c}^1,{\bf {\bar c}}^1,\dots,{\bf c}^k,{\bf {\bar c}}^k\right)$. 
The set $D$ we call the set of additional partitions.
And 
${\bf C}:=\sum_{i=1}^k\,\left(|{\bf c}^i|+|{\bf {\bar c}}^i|\right)$ denotes the total number of additional profiles
(where we exclude all profiles equal to $\Delta_0=1^d$).

We consider $H^{\textsc{e},\textsc{f}}_{d,N}$ (which are Hurwitz numbers if $N\ge d$, see (\ref{generating})) as functions 
of the set of partitions $\Delta^{(1)},\dots,\Delta^{(\textsc{f})} $, also the set of $b$ partitions $\Gamma=1^{d-2}2$
and of the set of the additional partitions $D\left({\bf c}^1,{\bf {\bar c}}^1,\dots,{\bf c}^k,{\bf {\bar c}}^k\right)$. 
Using the mixed notation we write
\[
 H^{\textsc{e},\textsc{f}+b+{\bf C}}_{d,N}
\left({\bf c}^1,{\bf {\bar c}}^1,\dots,{\bf c}^k,{\bf {\bar c}}^k,\Gamma_1,\dots,\Gamma_b,
\Delta^{(1)}\,\dots,\Delta^{(\textsc{f})}\right)
\]

We want to consider weighted sums of such Hurwitz numbers over occupation numbers of the additional partitions
${\bf c}^1,{\bf {\bar c}}^1,\dots,{\bf c}^k,{\bf {\bar c}}^k$.

Next we write down the corresponding weights.

{\bf Weight functions.}  We recall that we introduced partitions $\Delta^i_j,\Delta^{\bar i}_j$ and ${\bar i}=1,\dots,k$
which occur $c^i_j$ (respectively ${\bar c}^i_j$) times in the argument of $H^{\textsc{e},\textsc{f}+b+{\bf C}}_{d,N}$ 
(and $H^{\textsc{e},\textsc{f}+b+{\bf C}}_{d,N}$ are the Hurwitz numbers in case $N\ge d$). 

In what follows we shall need the set $i=1,\dots,k$ of the following functions of 3 parameters 
$n_i,q_i,t_i$.
\bea\label{W-i}
W_i({\bf c}^i): &=& W\left(n_i,q_i,t_i,{\bf c}^i\right)\,=\,(-1)^{|{\bf c}^i|}\,(-n_i)_{|{\bf c}^i|}\,
\prod_{j=1}^{d^*}\,\frac{\left(w_{\Delta^i_j}(q_i,t_i)\right)^{c^i_j}}{c^i_j!}
\\  
\label{W-i-*}
W_i^*({\bf {\bar c}}^i): &=& W\left(-n_i,0,t_i,{\bf{\bar c}}^i\right)\,=\,
(-1)^{|{\bf {\bar c}}^i|}\,(n_i)_{|{\bf {\bar c}}^i|}\,
\prod_{j=1}^{d^*}\,\frac{\left({w}_{\Delta^{\bar i}_j}(0,t_i) \right)^{{\bar c}^i_j}}{{\bar c}^i_j!}
\eea
where  $|{\bf c}^i|:=c^i_1+\dots +c^i_{d^*}$, 
$|{\bf {\bar c}}^i|:={\bar c}^i_1+\dots +{\bar c}^i_{d^*}$, and the notation $(n)_m:=n(n+1)\cdots (n+m-1)$ serves for 
the Pochhammer symbol. 
Functions $w_{\Delta^i_j}$ are defined as follows
\be
w_{\Delta^i_j}(q_i,t_i)=\,\left(\frac{1-t_i}{1-q_i}\right)^{d}\,
\prod_{s=1}^{\ell(\Delta^i_j)}\frac{1-q_i^{d^{ij}_s }}{1-t^{d^{ij}_s}}\,,\qquad
q_i,t_i\in\mathbb{C}
\ee
where $d^{ij}_1,\dots,d^{ij}_{\ell}$ are parts of the partition $\Delta^i_j$ and $\sum_{s=1}^{\ell} d^{ij}_s =d$.
\br\label{q=t}
As we see $w_\lambda(t,q)=w_\lambda(q,t)^{-1}$, in particular $w_\lambda(t,t)=1$.   
Also $w_{(1^d)}(q,t)=1$ for each $q,t$.
\er
\br
$w_\lambda(q^{-1},t)=w_\lambda(q,t^{-1})=(-1)^{d-\ell(\lambda)} w_\lambda(q,t)$
\er
\br\label{q=0,t=1}
As we see $w_\lambda(q,1)=\delta_{\lambda,1^d}$ for $q\neq 1$. It means that $W^*_i=\delta_{|{\bf{\bar c}}^i|,0}$ in case 
$q\neq 1,t=1$.
\er
\br
The multiplier $\frac{1}{c^i_j!}$ in (\ref{W-i}) corresponds to the permutation between identical profiles labeled by the 
partition $\Delta^i_j$.
\er

\vspace{2ex}

\paragraph{Weighted sums.}

We shall consider linear combinations of Hurwitz numbers related to different sets of partitions. 
To describe these sets we fix the weight of partitions, say, $d$. Then  we have a finite number of all different
partitions of this weight, we denote this number by $1+d^*$. Let us enumerate partitions from this set by
a subscript: $\Delta_j,\,\,j=1,\dots,d^*$. In what follows 
we need colored groups of partitions,
the color will be labeled by a superscript. Thus we have the set $\Delta^i_j$, $i=1,\dots,k$, $j=1,\dots,d^*$,
provided $\Delta^i_1$, $i=1,\dots,k$.
In our notation of functions defined on such sets, say, the Hurwitz number $H_\Omega$,
we shall write $H_\Omega(\{c^i_j),i=1,\dots,k,\, j=1,\dots,d^*\}$ instead of 
$H_\Omega\left(\{\Delta^i_j),i=1,\dots,k,\, j=1,\dots,d^*\}\right)$
replacing each partition $\Delta^i_j$ by the {\em occupation number} $c^i_j\ge 0$ which shows how many times 
the partition $\Delta^i_j$ occurs to be the argument of the function.
For us it is important to keep in mind the colored groups, and we shall use the following notation:
${\bf\Delta}^i=\Delta^i_1,\dots,\Delta^i_{d^*}$ and ${\bf c}^i(d)=c^i_1,\dots,c^i_{d^*}$. In what follows
we shall use two independent sets of occupation numbers: ${\bf c}^i(d)$ and 
${\bf {\bar c}}^i(d)={\bar c}^i_1,\dots,{\bar c}^i_{d^*}$. 

We consider  
\be\label{H-e-N-}
H^{\textsc{e},\textsc{f}+b+{\bf C}}_{d,N}\left(d; {\bf c}^1,{\bf {\bar c}}^1,\dots , {\bf c}^k,{\bf {\bar c}}^k,
\Gamma_1,\dots,\Gamma_b,
\Delta^{(1)},\dots,\Delta^{(\textsc{f})}\right)
\ee

the set of partitions $\{\Delta^i_j,\,j=1,\dots,d^*\}$ is the set of all different partitions of the weight $d$ except $1^d$

there are $k$ copies of such sets, each set is labeled by the superscript $i=1,\dots,k$

the number $c^i_j\ge 0$ (${\bar c}^i_j\ge 0$ says how many times the partition $\Delta^i_j$ (resp. $\Delta^{\bar i}_j$) 
occurs to appear among arguments of $H_{\textsc{e}}^{d,N}$

where $\Gamma_1=\cdots =\Gamma_b=1^{d-2}2$.  

Thus $H^\textsc{e}_{d,N}$ depends
on the set of $b+\textsc{f}+\sum_{i=1}^k \sum_{j=1}^{d^*}\left( c^i_j+{\bar c}^i_j\right)$ partitions.

We will study the weighted sums of $H^{\textsc{e}}_{d,N}$   
over the sets $c_j^i$ and ${\bar c}_j^i$ while partitions $\Delta^{(1)}\,\dots,\Delta^{(\textsc{f})}$ are fixed. 

\vspace{2ex}

Denote by
$$ L^{\textsc{e}}_{d,N}\left(d\,|\,\{n_i,q_i,t_i\}\,|\, b \,|\,
\Delta^{(1)},\dots,\Delta^{(\textsc{f})}\right)=$$

\be\label{L-E}
\sum_{{\bf c}^1,\dots , {\bf c}^k \atop {\bf {\bar c}}^1,\dots , {\bf {\bar c}}^k } \, 
H^{\textsc{e},\textsc{f}+b+{\bf C}}_{d,N}
\left(d; {\bf c}^1,{\bf{\bar c}}^1,\dots , {\bf c}^k,{\bf{\bar c}}^k,
\Gamma_1,\dots,\Gamma_b,
\Delta^{(1)},\dots,\Delta^{(\textsc{f})}\right)\,
\prod_{i=1}^k \, W_i({\bf c}^i) W_i^*({\bf {\bar c}}^i)
\ee
where both $W_i$ and $W_i^*$ depend on $\{n_i,q_i,t_i  \}$ are defined by (\ref{W-i}) and (\ref{W-i-*}).

\br
In case $d\le N$ and $\textsc{e}=1,2$ the number $H^\textsc{e}_{d,N}$ is the Hurwitz number. In what follows we will 
refer this case as Hurwitz case. The $3k$-parametric sum (\ref{L-E}) involves the Hurwitz number of 
$b+\textsc{f}+\sum_{i=1}^k \left(|{\bf c}^i|+ |{\bf {\bar c}}^i| \right)$ partitions where $b+\textsc{f}$
partitions are fixed and the summation ranges over the rest 
${\bf C}:=\sum_{i=1}^k \left(|{\bf c}^i|+ |{\bf {\bar c}}^i| \right)$ partitions. 
\er
 
 In case we choose $d\le N$ we obtain weighted sums of Hurwitz numbers $H_{\Omega}=H^{\textsc{e}}_{d,N}$ for the base
 surface with Euler characteristic $\textsc{e}$.

 \bp \label{tau=Hurwitz-sums}
The series (\ref{hyp-tau-e-f}) generates weighted sums $L^{\textsc{e}}_{d,N}$:
\[
 \tau^{(\textsc{e},\textsc{f})}\left(N,n,\bpow^{(1)},\dots,\bpow^{(\textsc{f})}|\{q_i,t_i,n_i \},\beta\right)
 =
\]
 \[
 \sum_{d,b=0}^\infty\,\frac{\beta^b}{b!}\,\sum_{\Delta^{(1)},\dots ,\Delta^{(\textsc{f})}}\, 
 L^{\textsc{e}}_{d,N}\left(\,\{n_i,q_i^{a_i}t_i^n,t_i\}\,|\, b \,|\,
\Delta^{(1)},\dots,\Delta^{(\textsc{f})}\right)\,
\prod_{i=1}^{\textsc{f}}\,\bpow^{(i)}_{\Delta^{(i)}}
\]

 \ep

\begin{proof}

Let us consider each factor in (\ref{hyp-tau-e-f}) related to a term lebaled by $\lambda$.

We start with the factors in (\ref{r}). 
First of all from (\ref{p-m-a-i-q-i})
\[
p_m(q_i,t_i):=\frac{1-q_i^{ m}}{1-t_i^m},\quad
p_m(0,t_i):=\frac{1}{1-t_i^m}
\]
for a partition $\Delta=(d_1,d_2,\dots)$ we obtain
\[
 \bpow_\Delta(q_i,t_i):=\left(p_1(q_i,t_1)\right)^{d_1} \left(p_2(q_i,t_1)\right)^{d_2} \cdots =
 \left( \frac{1-q_i}{1-t_i}\right)^{d_1} \left( \frac{1-q_i^{ 2}}{1-t_i^2}\right)^{d_2} \cdots
\]
Then from (\ref{char-map-introduction}) we obtain (\ref{ratio-of-s-s=ratio-of-chi}) and (\ref{w-introduction}).
Expanding (\ref{ratio-of-s-s=ratio-of-chi}) we obtain
\[
\sum_{}\,(-1)^{|{\bf c}^i|}\,(-n_i)_{|{\bf c}^i|}\,
\prod_{j=1}^{d^*}\,\frac{\left(w_{\Delta^i_j}(q_i,t_i)\right)^{c^i_j}}{c^i_j!}
\left(\frac{|C_{\Delta^{\bar i}_j}|\,\chi_\lambda(\Delta^i_j)}{\chi_\lambda(1)}\right)^{c^i_j}
\]
in enumerator and
\[  
\sum_{}\,(-1)^{|{\bf {\bar c}}^i|}\,(n_i)_{|{\bf {\bar c}}^i|}\,
\prod_{j=1}^{d^*}\,\frac{\left({w}_{\Delta^{\bar i}_j}(0,t_i) \right)^{{\bar c}^i_j}}{{\bar c}^i_j!}
\left(\frac{|C_{\Delta^{\bar i}_j}|\,\chi_\lambda(\Delta^{\bar i}_j)}{\chi_\lambda(1)}\right)^{{\bar c}^i_j}
\]
in the denominator. The summation ranges over all partitions of the weight $d=|\lambda|$, namely, over
all occupation numbers $c^i_j$ and ${\bar c}^i_j$.

Next, for the factor $e^{\beta f(\lambda,n)}$ in (\ref{r}) from (\ref{Okounkov-f}) we obtain
\[
 \sum_{b}\,\frac{\beta^b}{b!} \left(\frac{|C_\Gamma|\chi_\lambda(\Gamma)}{\chi_\lambda(1)} \right)^b
\]

For each of $\textsc{f}$ Schur functions in the product (\ref{hyp-tau-e-f}), we use (\ref{char-map-introduction}) to write
\[
s_\lambda(\bpow^{(j)})= \, \left(p^{(j)}_1\right)^d
\left(1+
\, \left(p^{(j)}_1\right)^{-d}\sum_{\Delta\neq 1^d}\,\,|C_{\Delta^{(j)}}|\,\,
\frac{\chi_\lambda(\Delta^{(j)})}{\chi_\lambda(1)}\,\,\bpow^{(j)}_{\Delta^{j}} \,
\right)\, \frac{\chi_\lambda(1)}{d!}   
\]

Finely we sum the product of all mentioned factors over $\lambda$, and using the defying relation 
(\ref{Hurwitz-counting})  for Hurwitz numbers we obtain the Proposition \ref{tau=Hurwitz-sums}.

\end{proof}

 Let us reduce $3k$-parametric families of weighted Hurwitz numbers as in examples below. 

{\bf Example 1.} We take $q_i=t_i,\,i=1,\dots,k$. (See Remark \ref{q=t}). Then we obtain $2k$-parametric family:
$$ L^{\textsc{e}}_d\left(\,b\,|\,\{n_i,t_i,t_i\}\,|\,
\Delta^{(1)},\dots,\Delta^{(\textsc{f})}\right)=$$

\[
\sum_{{\bf c}^1,\dots , {\bf c}^k \atop {\bf {\bar c}}^1,\dots , {\bf {\bar c}}^k } \, 
H_\Omega\left(d; {\bf c}^1,{\bf {\bar c}}^1,\dots , {\bf c}^k,{\bf {\bar c}}^k,
\Gamma_1,\dots,\Gamma_b,
\Delta^{(1)},\dots,\Delta^{(\textsc{f})}\right)\,W
\]
\[
W=
\prod_{i=1}^k \,\frac{
(-n_i)_{|{\bf c}^i|}\,(n_i)_{|{\bf {\bar c}}^i|}}{(-1)^{|{\bf c}^i|+|{\bf {\bar c}}^i|}}\,
\prod_{j=1}^{d^*}\,\frac{1}{c^i_j!}
\prod_{j=1}^{d^*}\,\frac{\left({w}_{\Delta^{\bar i}_j}(0,t_i) \right)^{{\bar c}^i_j}}{{\bar c}^i_j!}
\]
Take further $n_i=-1,\,i=1,\dots,k$. Due to the factor $(n_i)_{|{\bf {\bar c}}^i|}$  there is a single profile in each
color group labeled by bar. After some changing of notations under summation we obtain

\bea\label{L-E-Ex1}
\sum_{{\bf c}^1,\dots , {\bf c}^k \atop \Delta^1,\dots , \Delta^k } \, 
H_\Omega\left(d; {\bf c}^1,\Delta^1,\dots , {\bf c}^k,\Delta^k,
\Gamma_1,\dots,\Gamma_b,
\Delta^{(1)},\dots,\Delta^{(\textsc{f})}\right)\,W\,
\\
W=-\left(\prod_{i=1}^k 
\frac{
{|{\bf c}^i|!}}{(-1)^{|{\bf c}^i|}}\,
\prod_{j=1}^{d^*}\,\frac{1}{c^i_j!}
\right)
\prod_{i=1}^k \,
{w}_{\Delta^{i}}(0,t_i) 
\eea
where
\be\label{w:q=t^2}
w_{\Delta^i}(0,t_i)=\,\left(1-t_i\right)^{d}\,
\prod_{s=1}^{\ell(\Delta^i)}\frac{1}{1-t_i^{d^{i}_s} }\,,\qquad
\ee
where $\Delta^i=(d^i_1,d^i_2,\dots),\,i=1,\dots,k$.

{\bf Example 2.}
The sum (\ref{L-E}) has an interesting limit in case $q_i\to 1$ for all $i$. In this case thanks to Remark \ref{q=0,t=1} 
the terms where any of ${\bar c}^i_j$ is nonvanishing vanishes. Then, denoting
\[ \lim_{q_i\to 1\atop i=1,\dots,k} L^\textsc{e}_d \left(\,b\,|\,\{n_i,q_i^{a_i},q_i\}\,|\,
\Delta^{(1)},\dots,\Delta^{(\textsc{f})}\right)=\texttt{L}_{\textsc{e}}\left(d\,|\,b\,|\,\{n_i,a_i \}\,|\,
\Delta^{(1)},\dots,\Delta^{(\textsc{f})}\right), \] 
we obtain $2k$-parametric family
$$
\texttt{L}^{\textsc{e}}_d\left(\,b\,|\,\{n_i,a_i \}\,|\,
\Delta^{(1)},\dots,\Delta^{(\textsc{f})}\right)=
$$

\be\label{L-E-Ex2}
\sum_{{\bf c}^1,\dots , {\bf c}^k  } \, 
H_\Omega\left(d; {\bf c}^1,\dots , {\bf c}^k,
\Gamma_1,\dots,\Gamma_b,
\Delta^{(1)},\dots,\Delta^{(\textsc{f})}\right)\,
\prod_{i=1}^k \, \frac{(-1)^{|{\bf c}^i|}}{a_i^{d|{\bf c}^i|}}(-n_i)_{|{\bf c}^i|}\,
\prod_{j=1}^{d^*}\,\frac{a_i^{\ell(\Delta^i_j)c^i_j}}{c^i_j!}
\ee
Let us note that the exponent $\sum_{j=1}^{d^*} \ell(\Delta^i_j)c^i_j$ has the meaning of the sum of lengths of all 
partitions of the color group $i$ excluding thye partition $1^d$, while $|{\bf c}^i|:=\sum_{j=1}^{d^*}$ has the meaning of 
total number of these partitions in this color group.

If we put $\textsc{e}=\textsc{f}=2$ and $n_i =\pm 1$ in (\ref{L-E-Ex2}) we obtain the case considered in \cite{HO-2014}.

{\bf Example 3.}
We can further reduce (\ref{L-E-Ex2}) putting $n_i=1$ for $i=1,\dots, k$. In this case each $(-n_i)_{|{\bf c}^i|}=0$
until $|{\bf c}^i|=1$. It means that a single partition, $\Delta^i_{j_i}$, in each color group $i$, contributes to the set of
the arguments of $H^\textsc{e}_{d,N}$, and the summation in (\ref{L-E-Ex2}) is the summation over the set of various 
$j_1,\dots,j_{k}$ where $1\le j_i \le d^*$.
Instead of (\ref{L-E-Ex2}) we obtain $k$-parametric family
$$
\texttt{L}^{\textsc{e}}_d\left( b\,|\,\{1, a_i\}\,|\,
\Delta^{(1)},\dots,\Delta^{(\textsc{f})}\right)=
$$

\be\label{L-E-Ex3}
\sum_{1\le j_1,\dots,j_k\le d^* } \, 
H_\Omega\left(d; \Delta^1_{j_1},\dots , \Delta^k_{j_k},
\Gamma_1,\dots,\Gamma_b,
\Delta^{(1)},\dots,\Delta^{(\textsc{f})}\right)\,
\prod_{i=1}^k \, 
\,a_i^{\ell(\Delta^i_{j_i})-d}
\ee

Relations (\ref{L-E-Ex2})-(\ref{L-E-Ex3}) generates sums of Hurwitz numbers with certain conditions upon
the partitions lengths. For instance, from (\ref{L-E-Ex3}) we obtain
$$
\res_{a_1=0} a_i^{d-l_1-1}\cdots \res_{a_k=0} a_k^{d-l_k-1} \, \texttt{L}^{\textsc{e}}_d\left(\,b\,|\,\{1, a_i\}\,|\,
\Delta^{(1)},\dots,\Delta^{(\textsc{f})}\right) =
$$

\be\label{L-E-Ex3-res}
\sum  H_\Omega(d,\Gamma_1,\dots,\Gamma_b,
\Delta^1,\dots,
\Delta^k,
\Delta^{(1)},\dots,\Delta^{(\textsc{f})})
\ee
where the sum is taken over all partitions $\Delta^i$ of the fixed lengths $l_i$.

{\bf Example 4.} One can re-write (\ref{L-E-Ex2}) as different $2k$-parametric family as follows. 
Using
\[
(k)_{|{\bf c}^i|}=\frac{\Gamma(|{\bf c}^i|+k)}{\Gamma(k)},\quad
 (1-x_i)^{-|{\bf c}^i|} =\,1+\sum_{k=1}^\infty\, 
 \frac{x_i^k}{k!}\frac{\Gamma(|{\bf c}^i|+k)}{\Gamma(|{\bf c}^i|)}=
 1+\sum_{k=1}^\infty\, 
 \frac{x_i^k}{k!}\frac{(k)_{|{\bf c}^i|}\Gamma(k)}{\Gamma(|{\bf c}^i|)}
\]
and taking the sum of (\ref{L-E-Ex2}) over $n_i=-1,-2,\dots$ with the weight $-\frac{x_i^{-n_i}}{n_i}$ ,
$x_i=1-z_i$, we obtain 
$$
(-1)^k\sum_{n_1,\dots,n_k < 0}\,\prod_{i=1}^k\frac{(1-z_i)^{-n_i}}{n_i}\,\texttt{L}^{\textsc{e}}_d
\left(\,b\,|\,\{n_i,a_i \}\,|\,
\Delta^{(1)},\dots,\Delta^{(\textsc{f})}\right)=
$$

\be\label{L-E-Ex4}
\sum_{{\bf c}^1,\dots , {\bf c}^k  } \, 
H_\Omega\left(d; {\bf c}^1,\dots , {\bf c}^k,
\Gamma_1,\dots,\Gamma_b,
\Delta^{(1)},\dots,\Delta^{(\textsc{f})}\right)\,\prod_{i=1}^k\,\left(z_i^{-|{\bf c}^i|}-1\right)\,
(-1)^{|{\bf c}^i|}
\Gamma(|{\bf c}^i|)\,
\prod_{j=1}^{d^*}\,\frac{a_i^{\ell(\Delta^i_j)c^i_j}}{c^i_j!}
\ee
where $z_i$ and $a_i$ , $i=1,\dots,k$ are free parameters. Picking up terms at given powers of $z_i$ and $a_i$
one can obtain the sums of Hurwitz numbers under the conditions: the number of profiles in each color group, $|{\bf c}^i|$, 
is fixed and the sum of profile lenghts inside each color group, $\sum_{j=1}^{d^*}\ell(\Delta^i_j)c^i_j$, is also fixed, see
(\ref{T-BKP}),(\ref{T-KP}) below.
Let us note, that the term related to $|{\bf c}^i|=0$ should be treated as  
$\lim_{|{\bf c}^i| \to 0}\left(z_i^{-|{\bf c}^i|}-1\right)\,\Gamma(|{\bf c}^i|)=-\log z$.

\section{Toda lattice and BKP tau functions \label{BKP-section}}

\subsection{Pochhammer symbols and content products}
For a given partition $\lambda=(\lambda_1,\dots,\lambda_l)$ and a function on the one-dimensional lattice
$r(x), \, x\in \mathbb{Z}$, we introduce the generalized Pochhammer symbol $r_n(x)$ as
\be
r_n(x)\,=\,r(x)r(x+1)\cdots r(x+n-1)
\ee
and the generalized Pochhammer symbol, $r_\lambda$, related to a partition $\lambda$ as
\be
r_\lambda(x)\,=\,r_{\lambda_1}(x)r_{\lambda_2}(x-1)\cdots r_{\lambda_l}(x-l+1)
\ee
It may be written also as a content product as follows
\[
 r_\lambda(x)\,=\,\prod_{i,j\in\lambda} r(x+j-i)
\]
where $j-i$ is a content of the node in $i$-th row and $j$-th column of the Young diagram of a partition $\lambda$,
see \cite{OS-TMP} for more details.

The content product plays an important role in many combinatorial problems, see for instance \cite{GJ} where links to
integrable systems were also pointed out.

\br\label{fg}
 (1) If $r=fg$, then $r_\lambda(x)=f_\lambda(x)g_\lambda(x)$. (2) If ${\tilde r}(x)=\left(r(x)\right)^n,\,n\in\mathbb{C}$,
then ${\tilde r}_\lambda(x)=
\left({r}_\lambda(x)\right)^n$.
\er

{\bf Example I}. For $r(x)=x$ the generalized Pochhammer symbol coincides with the familiar one:
\[
 r_\lambda(x)=(x)_\lambda\,,\quad
 (x)_\lambda :=(x)_{\lambda_1}(x-1)_{\lambda_1} \cdots (x-l+1)_{\lambda_l}\,,
 \quad (x)_n=\frac{\Gamma(x+n)}{\Gamma(x)}
\]
If we take $r(x)=x^{n},\,n\in\mathbb{C}$, then $r_\lambda(x)=\left( (x)_\lambda \right)^n$

{\bf Example II}.
For $r(x)=1-qt^x$ the generalized Pochhammer symbol coincides with the $q$-deformed one 
(for some reasons we replace the letter $q$ by the letter $t$):
\[
r_\lambda(x)=(qt^x;t)_\lambda ,\quad
 (qt^x;t)_\lambda :=(qt^x;t)_{\lambda_1}(qt^{x-1};t)_{\lambda_1} \cdots (qt^{x-l+1};t)_{\lambda_l}\,,
 \quad (qt^x;t)_n=(1-qt^x)\cdots (1-qt^{x+n-1})
\]
This case may be also referred as trigonometric. The trigonometric $r$ may be viewed as the infinite product of
$r$ from the Example I.
If we take $r(x)=\left(1-qt^x \right)^{n},\,n\in\mathbb{C}$, then $r_\lambda(x)=\left( (qt^x;t)_\lambda \right)^n$,
see Remark \ref{fg}.

{\bf Example III}.
For $r(x)=\theta(cx,t)$, where $\theta(cx,t)$ is the Jacoby theta function, 
\[
 \theta(cx,\tau):=\sum_{n\in\mathbb{Z}}\exp (\pi i n^2 \tau + 2c\pi i n x)=
 (q;q)_\infty
\prod_{n=1}^\infty \left(1 + q^{n-{\frac 12}} t^{ x}\right)\left(1 + 
q^{n-{\frac 12}} t^{- x}\right) 
\]
where $q=e^{2\pi i \tau},\, t=e^{2c\pi i}$. As one can note the elliptic $r$ is the infinite product
of ``trigonometric'' $r$ from the Example II.
The generalized Pochhammer symbol is then the elliptic Pochhammer symbol:
\[
r_\lambda(x)=[xc;t]_\lambda ,\quad
 [cx;t]_\lambda :=[cx;t]_{\lambda_1}[c(x-1);t]_{\lambda_1} \cdots [c(x-l+1);t]_{\lambda_l}\,,
 \quad [cx;t]_n:=\theta(cx|t)\cdots \theta(c(x+n-1)|t)
\]
If we take $r(x)=\left( \theta(cx,t) \right)^{n},\,n\in\mathbb{C}$, then $r_\lambda(x)=\left( [xc;t]_\lambda \right)^n$

The shall refer the cases $r(x)=a+cx,\,r(x)=1-qt^x,\,r(x)=\theta(cx,\tau)$ as respectively rational, trigonometric and 
elliptic ones.

{\bf Example IV}. For $r(x)=e^{\beta x}$ we obtain the case considered in \cite{Okounkov-2000}
\[
 r_\lambda(x)=e^{\beta f_2(\lambda,x)}
\]
where
\[
f_2(\lambda,x)\,=\,\frac 12 \,\sum_{i} \left[(x+\lambda_i-i+\frac12)^2 - (x-i+\frac 12)^2\right],
\qquad
 f_2(\lambda,0)\,+\,x|\lambda|
\]
In \cite{Okounkov-2000} it was shown that $f_2$ may be written in form  
 (\ref{Okounkov-f}) and this is a key relation to link the generating function for Hurwitz numbers 
 to integrable systems.

\vspace{2ex}

\paragraph{Examples of the parametrizations of $r$}
It may be suitable to choose the function $r$ as the product of
functions $r^{(1)}\cdots r^{(k)}$ (then $r_\lambda= r^{(1)}_\lambda\cdots r^{(k)}_\lambda$, see Remark \ref{fg}), where 
each one is parametrized by the set of parameters $\beta^{(i)}=(\beta^{(i)}_1,\beta^{(i)}_2,\dots)$
in one of three ways:
\bea\label{r-for-x-rational-Pochhamer}
({\rm I})\, :\qquad\qquad\qquad\qquad\qquad
r^{(i)}(x)\,&=&\,e^{\beta_0^{(i)} +\sum_{m > 0} \,   \frac1m \beta_m^{(i)} x^m}
\\
\label{r-for-t_i-x-trig-Pochhamer}
({\rm II})\, :\qquad\qquad\qquad\qquad\qquad
r^{(i)}(x)\,&=&\,e^{\sum_{m > 0} \,   \frac1m \beta_m^{(i)} t^{mx}_i}
\\
\label{r-for-theta-Pochhamer}
({\rm III})\, :\qquad\qquad\qquad\qquad\qquad
r^{(i)}(x)\,&=&\,e^{\sum_{m > 0} \, \frac 1m \beta_m^{(i)} \left(e^{mc_i(a_i+x)}+e^{-mc_i(a_i+x)} \right)}
\eea
To obtain respectively the usual, the trigonometric and elliptic Pochhammer symbols we
choose
\bea\label{beta-r-for-x-rational-Pochhamer}
\beta^{(i)}_m \,&=& \,n_i\, (-a_i)^{-m}\,,\quad \beta^{(i)}_0= \,n_i\, \log a_i
\\
\label{beta-for_i-x-trig-Pochhamer}
 \beta_m^{(i)}  \,&=& \, -n_i q_i^{m}
\\
\label{beta-for-theta-Pochhamer}
 \beta_m^{(i)}  \,&=& \, n_i(-1)^m\frac{q_i^m}{1-q_i^m}
\eea
Then respectively we obtain
\be\label{r-lambda-for-x-rational-Pochhamer}
r_\lambda(x)\,=\,\prod_{i=1}^k\, \left( (a_i+n)_\lambda \right)^{n_i}
\ee
\be\label{r-lambda-for_i-x-trig}
r_\lambda(x)\,=\,\prod_{i=1}^k\,\left((q_it_i^n;t_i)_\lambda \right)^{n_i}
\ee
\be\label{r-lambda-for-theta-Pochhamer}
r_\lambda(x)\,=\,\prod_{i=1}^k\,\left([c_i(a_i+x),q_i]_\lambda\right)^{n_i}
\ee
where $ a_i,q_i,t_i,n_i $ are complex numbers.

Now let us consider limiting expressions of rational, trigonometric and ellitpic cases
\bea\label{limit-rac}
 r^{(i)}(x)\,&=&\, \left(1+\frac{x}{n_ia_i}\right)^{n_i} \,\,\to \,\exp \,{\frac{x}{a_i}}
\\
\label{limit-trig}
 r^{(i)}(x)\,&=&\,\left(1+\frac{q_i}{n_i} t_i^x\right)^{n_i}\,\, \to \, \exp \, {q_it_i^x}
\\
\label{limit-elliptic}
 r^{(i)}(x)\,&=&\,\left(\theta(a_i,\tau_i )\right)^{n_i} \,\qquad \to \, \exp \, {q_i^{-\frac12}(t_i^x+t_i^{-x}})
\eea
Formula similar to (\ref{limit-rac}) was previously obtained in \cite{HarnadMathieu-sept-2014}. It result in
the choice of $r$ as in Example IV above (where we put $\beta=\frac 1a$),
and allows to view
the Okounkov tau function \cite{Okounkov-2000} as a limit $n\to\infty$ of the certain hypergeometric function of matrix 
arguments \cite{GRich} of type ${_{n+1}F}_0$.


Let us note that the variables $\beta$ in the parametrization (\ref{r-for-theta-Pochhamer}) is a triangle transform of 
variables ${t^*}$
introduced in \cite{O-2003}. 
This follows from the relation $r(x)=e^{U_{x-1}-U_x}$, where $U_x=\sum_{m\neq 0} \,t^*_m t^{mx}$.

It actually means that $\tau_r^{\rm 2KP}(N,n,\bpow^{(1)},\bpow^{(2)})$ is a 2KP tau function in variables 
$(n,\bpow^{(1)},\beta)$ in case $p^{(2)}_m=\frac{1}{1-t^m}$, see details in \cite{HO-future}.

At last we need the important
\bl
\be
s_\lambda(\bpow(q,t))\,=\,(q;t)_\lambda\,s_\lambda(\bpow(0,t)
\ee
and its limiting case
\be
s_\lambda(\bpow(a))\,=\,(a)_\lambda\,s_\lambda(\bpow_\infty)
\ee
where $\bpow(q,t)$ is defined by 
\be\label{p-m-q-t}
\bpow(q,t)=\left(p_1(q,t),p_2(q,t),\dots \right)\,,\quad
p_m(q,t):=\frac{1-q^{ m}}{1-t^m}
\ee
$\bpow_\infty=(1,0,0,\dots)$ and
\be\label{bpow(a)}
\bpow(a)\,=\,(a,a,\dots )
\ee
\el
which may be easily obtained from the consideration in Ch I of \cite{Mac}. This Lemma allows to identify series
(\ref{hyp-tau-e-f}) where $\textsc{e}=\textsc{f}=1,2$ with BKP and 2KP tau functions \cite{OS-2000},\cite{OST-I}.

\subsection{TL tau function and TL hypergeometric tau function.} 
Here we recall few facts about the Toda lattice rau functions,
and for details we  refer to the paper \cite{UT}.
The simplest Hirota equation for the Toda lattice is
\be\label{the-first-TL-Hirota}
\frac{\partial^2 \tau^{\rm TL}\left(n,\bpow, \bbpow \right)}{\partial p_1\partial {\bar p}_1}
\tau^{\rm TL}\left(n,\bpow, \bbpow \right) -
\frac{\partial \tau^{\rm TL}\left(n,\bpow, \bbpow \right)}{\partial p_1}
\frac{\partial\tau^{\rm TL}\left(n,\bpow, \bbpow \right)}{ \partial {\bar p}_1} =
-\tau^{\rm TL}\left(n+1,\bpow, \bbpow \right) \tau^{\rm TL}\left(n-1,\bpow, \bbpow \right)
\ee
TL tau function may be written in form
\be\label{TL-tau-general}
 \tau^{\rm TL}\left(n,\bpow, \bbpow \right)=\sum_{\lambda\in\Pa}\,s_\lambda(\bpow)\,g_{\lambda,\mu}(n) s_\mu(\bbpow)
\ee
where $g_{\lambda,\mu}$ is a determinant, see \cite{TakasakiSchur}. The Schur function is defined as follows
\be\label{Schur-pol}
s_\lambda(\bpow) \,=\, \det \left( s_{\lambda_i-i+j}(\bpow) \right)_{i,j}\,,\quad
e^{\sum_{m>0} \, \frac 1m \,z^m\, p_m} \,=: \sum_{m\ge 0} \,z^m\, s_m(\bpow)
\ee

We shall denote the length of a partition $\lambda$ by $\ell(\lambda)$ and the weight of $\lambda$ by $|\lambda|$,
see \cite{Mac}.

\paragraph{TL tau function of the hypergeometric type.} These are
\be\label{TL-sums}
g(n)\sum_{\lambda\in \Pa } \, r_\lambda(n)\, s_\lambda(\bpow)\, s_\lambda(\bbpow) \,=
: \tau_r^{\rm TL}(n,\bpow)
\ee
In case $r$ vanishes at certain site number $M$ it is better to speak about

\paragraph{Semi-infinite TL tau function} with the origin at the site number $M$. These are
\be\label{semiinfinteTL-sums}
g(n)\sum_{\lambda\in \Pa \atop \ell(\lambda)\le n-M} \, r_\lambda(n)\, s_\lambda(\bpow)\, s_\lambda(\bbpow) \,=
: \tau_r^{\rm TL}(M,n,\bpow)
\ee
This tau function is a particular case of the previous one if we choose $r(M)=0$.

In \cite{OS-2000} there were written down two main examples of such tau functions:
\be\label{hyper-matrix-arg-RichGross}
\tau_r(n,\bpow^{(1)},\bpow^{(2)})=\sum_\lambda\, \left( s_\lambda(\bpow_\infty) \right)^{q-p}\,
\frac{\prod_{i=1}^p\,s_\lambda(\bpow(a_i+n))}
{\prod_{i=1}^q\,s_\lambda(\bpow(b_i+n)) }\,
s_\lambda(\bpow^{(1)}) s_\lambda(\bpow^{(2)}) 
\ee
which may be related to the hypergeometric function of matrix argument \cite{GRich},
and
\be\label{Milne=TL}
\tau_r(n,\bpow^{(1)},\bpow^{(2)})=\sum_\lambda\, \left( s_\lambda(\bpow_\infty) \right)^{q-p}\,
\frac{\prod_{i=1}^p\,s_\lambda(\bpow(t^{a_i+n},t))}{\prod_{i=1}^q\,s_\lambda(\bpow(t^{b_i+n},t))}\,
s_\lambda(\bpow^{(1)}) s_\lambda(\bpow^{(2)}) 
\ee
which which may be related to the Milne's hypergeometric functions \cite{Milne}.

\vspace{2ex}

\subsection{BKP tau function.}
We are
interested in a certain subclass of the BKP tau functions (\ref{BKP-tau-general}) written down in \cite{OST-I} and called
BKP hypergeometric tau functions, and also in the similar class of TL tau functions (\ref{TL-tau-general}) found in
\cite{KMMM}, \cite{OS-2000}.

There are two different BKP hierarchies of integrable equations, one was introduced by the Kyoto group in \cite{JM}, the other was
introduced by V. Kac and J. van de Leur in \cite{KvdLbispec}. We need the last one. This hierarchy includes the celebrated
KP one as a particular reduction.  In a
certain way (see \cite{LeurO-2014}) the BKP hierarchy may be related to the three-component KP hierarchy introduced in \cite{JM}
(earlier described in \cite{ZakharovShabat} with the help of L-A pairs with matrix valued coefficients). For a detailed
description of the BKP we refer readers to the original work \cite{KvdLbispec}, and here
we write down the first non-trivial equations for the BKP tau function (Hirota equations). These are
\bea\label{Hirota-elementary-1}
\frac 12 \frac{\partial\tau(N,n,\bpow)}{\partial p_2} \tau(N+1,n,\bpow)-
\frac 12 \tau(N,n,\bpow)\frac{\partial\tau(N+1,n,\bpow)}{\partial p_2} \nonumber
+\frac 12 \frac{\partial^2\tau(N,n,\bpow)}{\partial^2 p_1} \tau(N+1,n,\bpow)\\
+\frac 12 \tau(N,n,\bpow)\frac{\partial^2\tau(N+1,n,\bpow)}{\partial^2 p_1}
- \frac{\partial\tau(N,n,\bpow)}{\partial p_1}\frac{\partial\tau(N+1,n,\bpow)}{\partial p_1}
=\tau(N+2,n,\bpow)\tau(N-1,n,\bpow)
\eea
\bea\label{Hirota-elementary-2}
\frac 12 \tau(N,n+1,\bpow)\frac{\partial^2 \tau(N+1,n,\bpow)}{\partial^2 p_1}-
\frac 12 \frac{\tau(N,n+1,\bpow)}{\partial^2 p_1} \tau(N+1,n,\bpow)=\nonumber
\\
\frac{\partial\tau(N+2,n,\bpow)}{\partial p_1}\tau(N-1,n+1,\bpow)-
 \frac{\partial \tau(N+1,n+1,\bpow)}{\partial p_1}\tau(N,n,\bpow)
\eea
The BKP tau functions depend on the set of higher times $t_m=\frac 1m p_m$, $m>1$ and the discrete parameter $N$.
In \cite{OST-I} the second discrete parameter $n$ was added and equation (\ref{Hirota-elementary-2}) relates BKP tau
functions with neighboring $n$. The complete set of the Hirota equations with two discrete parameters is written down in the
Appendix.

The general solution to Hirota equations may be written as
\be\label{BKP-tau-general}
\tau^{BKP}\left(N,n,\bpow \right)=\sum_{\lambda\in\Pa}\,A_\lambda(N,n) s_\lambda(\bpow)
\ee
where $\Pa$ is the set of all partitions and
where $A_\lambda$ solves Plucker relations for isotropic Grassmannian and may be written in a pfaffian form.

\paragraph{BKP tau function of the hypergeometric type}
We consider sums over partitions of form
\be\label{BKP-sums}
g(n)\sum_{\lambda\in \Pa \atop \ell(\lambda)\le N} \, r_\lambda(n)\, s_\lambda(\bpow) \,=
: \tau_r^{\rm BKP}(N,n,\bpow)
\ee
where $\Pa$ is the set of all partitions, $s_\lambda$ are the Schur functions \cite{Mac} and the semi-infinite
set $\bpow=(p_1,p_2,\dots )$ is related to the called higher times in the soliton theory
${\bf t}=(t_1,t_2,\dots )$ via $p_m=mt_m$.
The constant $g(n)$ is not important and may be found in Appendix \ref{fermionic-appendix}, see (\ref{r-U}),(\ref{g(n)}).

\br\label{scaling}
For hypergeomtric tau functions (\ref{TL-sums}) and (\ref{BKP-sums}) there is an obvious transformation
$r_\lambda \to a^{-|\lambda|}r_\lambda,\,p^{1}_m \to ap^{(1)}_m,\,m>0$, which does not change the tau functions.
\er

\paragraph{Remarks.} According to \cite{OS-2000}, \cite{OST-I}

\vspace{2ex}

(1)
$\tau_r^{\rm TL}$ solves certain linear equation generalizing Gauss equation for Gauss hypergeometric function
which may be also referred as a ``string equation''

\vspace{2ex}

(2)
There are various determinantal formulae to present $\tau_r^{\rm TL}$, and pfaffian formulae to present
$\tau_r^{\rm BKP}$

\vspace{2ex}

(3)
Both $\tau_r^{\rm TL}$ and $\tau_r^{\rm BKP}$ may be obtained by the action of vertex operators on certain simple
functions

\vspace{2ex}

Also \cite{OST-I , HO-convolution , OShiota-2004} :

\vspace{2ex}

(4)
Sums (\ref{BKP-sums}) and (\ref{TL-sums}) may be considered as partition functions for models of random partitions
where a partition $\lambda$ contributes the weights $r_\lambda(n) \, s_\lambda(\bpow)$, or
$r_\lambda(n) \,s_\lambda(\bpow) \,s_\lambda(\bbpow)$ respectively

\vspace{2ex}

(5)
For certain specifications of $r$ and $\bpow$ sums (\ref{BKP-sums}) and (\ref{TL-sums}) may be viewed as multisoliton
tau functions. Similarly, for the same specifications,
they may be viewed as discrete versions of matrix models

\section{Generating Hurwitz series and the BKP and 2KP tau functions \label{Hurwitz-BKP-section}}

As we have seen
\[
\tau^{\textsc{e},\textsc{f}}\left(N,n,\beta,\{n_i,q_i,t_i\}\,|\,\bpow^{(1)},\dots,\bpow^{(\textsc{f})}\right)=
\]
\[
 \sum_{\lambda\atop \ell(\lambda)\le N}\,r^{{\bf n},{\bf q},{\bf t}}_\lambda(n)\prod_{i=1}^{\textsc{f}}s_\lambda(\bpow^{(i)})=
\]
\[
\sum_{d\ge 0}\sum_{\Delta^{(1)},\dots,\Delta^{(\textsc{f})}\in\Pa\atop
|\Delta^{(j)}|=d,\,j=1,\dots,\textsc{f}} 
\frac{\beta^b}{b!}
L^{\textsc{e},\textsc{f}}_{d,N}\left(b\,|\,\{n_i,q_it_i^n,t_i\}\,|\,
\Delta^{(1)},\dots,\Delta^{(\textsc{f})}\right)\prod_{j=1}^\textsc{f}\,
\bpow^{j}_{\Delta^{(j)}}
\]
where
\[
 r^{{\bf n},{\bf q},{\bf t}}_\lambda(n)\,=\,\left(qe^{\beta n}\right)^{|\lambda|}e^{\beta f_2(\lambda)}\,\prod_{i=1}^k\,
\left(\frac{(q_it_i^n;t_i)_\lambda}{\left(1-q_i\right)^{d}}\right)^{n_i}
\]
and where $
L^{\textsc{e},\textsc{f}}_{d,N}\left(b\,|\,\{n_i,q_it_i^n,t_i\}\,|\,
\Delta^{(1)},\dots,\Delta^{(\textsc{f})}\right)$ for $N\ge d$ is the weighted sum of Hurwitz numbers for the 
$d$-fold covering of the surface of Euler characteristic $\textsc{e}$ with fixed ramification profiles 
$\Delta^{(1)},\dots,\Delta^{(\textsc{f})}$ and $b$ profiles $\Gamma=1^{d-2}2$, and summation runs over additional profiles 
as it was described in Section \ref{Weighted-sums-Section}.

Let $\textsc{e}=\textsc{f}$.
We have
\begin{Theorem} \label{Th-BKP} 
For any complex numbers $\{n_i,q_i,t_i\}$ and integers 
$n$ and $N\ge 0$
the generating Hurwitz series 
$$\tau^{1,1}\left(N,n,\beta,\{n_i,q_i^{a_i},q_i\}\,|\,\bpow \right)$$ 
\[
=\,\sum_{d\ge 0}\sum_{\Delta\in\Pa\atop
|\Delta|=d}\, 
\frac{\beta^b}{b!}\,
L^{1,1}_{d,N}\left(b\,|\,\{n_i,q_it_i^n,t_i\}\,|\,
\Delta\right)\,
\bpow_{\Delta}
\]
\[
= \sum_{\lambda\atop \ell(\lambda)\le N}\,
r^{{\bf n},{\bf q},{\bf t}}_\lambda(n)\,s_\lambda(\bpow)
\]
where
\be
 r^{{\bf n},{\bf q},{\bf t}}_\lambda(n)\,=\,\left(qe^{\beta n}\right)^{|\lambda|}e^{\beta f_2(\lambda)}\,\prod_{i=1}^k\,
\left(\frac{(q_it_i^n;t_i)_\lambda}{\left(1-q_i\right)^{d}}\right)^{n_i}
\ee
 is a $\tau-$function of the BKP 
hierarchy where $\bpow=\left(p_1,p_2,\dots \right)$ plays the role of higher times. 
\end{Theorem}

\begin{Theorem} \label{Th-KP} 
For any complex numbers $\{n_i,q_i,t_i\}$ and integers 
$n$ and $N\ge 0$
the generating Hurwitz series 
$$\tau^{2,2}\left(N,n,\beta,\{n_i,q_i^{a_i},q_i\}\,|\,\bpow^{(1)},\bpow^{(2)}\right)$$ 
\[
=\,\sum_{d\ge 0}\sum_{\Delta^{(1)},\Delta^{(2)}\in\Pa\atop
|\Delta^{(j)}|=d,\,j=1,2}\, 
\frac{\beta^b}{b!}\,
L^{2,2}_{d,N}\left(b\,|\,\{n_i,q_it_i^n,t_i\}\,|\,
\Delta^{(1)},\Delta^{(2)}\right)\prod_{j=1,2}\,
\bpow^{j}_{\Delta^{(j)}}
\]
\[
= \sum_{\lambda\atop \ell(\lambda)\le N}\,
r^{{\bf n},{\bf q},{\bf t}}_\lambda(n)\prod_{i=1,2}s_\lambda(\bpow^{(i)})
\]
where
\be
 r^{{\bf n},{\bf q},{\bf t}}_\lambda(n)\,=\,\left(qe^{\beta n}\right)^{|\lambda|}e^{\beta f_2(\lambda)}\,\prod_{i=1}^k\,
\left(\frac{(q_it_i^n;t_i)_\lambda}{\left(1-q_i\right)^{d}}\right)^{n_i}
\ee
 is a $\tau-$function of the 2-component BKP 
hierarchy and also of the semiinfinite TL hierarchy where $\bpow^{(j)}=\left(p^{(j)}_1,p^{(j)}_2,\dots\right)$ play the role 
of higher times. 
\end{Theorem}

\begin{proof}. The statements of the Theorems directly follow from the consideration in the previous 
sections, see Proposition \ref{tau=Hurwitz-sums} and formulae (\ref{BKP-sums}) and (\ref{TL-sums}) where we choose
\be
r(x)\,=\,e^{\beta x}\,\prod_{i=1}^k\,\left(1-q_it_i^x \right)^{n_i}
\ee
This choice provides
\be
r_\lambda(x)\,=\,e^{\beta f_2(x)}\,\prod_{i=1}^k\,\left((q_it_i^x;t_i)_\lambda \right)^{n_i}
\ee
\end{proof}

Let us mark that if we introduce 
\be\label{}
F({\bf q,t,n};\beta, x)=\sum_{m>0}\,e^{\beta m}\,\prod_{i=1}^k\,\left(1-q_it_i^{m} \right)^{n_i}\,x^m
\ee
then
\[
\tau^{2,2}\left(N,n,\beta,\{n_i,q_i^{a_i},q_i\}\,|\,X,Y\right)\,
=\,\det\,\left[ F({\bf q,t,n};\beta, x^{(1)}_i x^{(2)}_j) \right]_{i,j=1,\dots,N}
\,,\quad  p^{(j)}_m=\sum_{i=1}^N\,\left(  x^{(j)}_i \right)^m
\]

\subsection{Particular cases \label{Particular cases}}

We consider the sums of Hurwitz numbers which correspond to the following options:
\begin{itemize}
\item $b,d,k\in\mathbb{Z}$, where $d>0$, $b,k\geq 0$;
\item integers $N_1,\dots,N_k \ge 0$ ;
\item  integers $L_1,\dots,L_k\ge 0$ ;
\item integers $c_j^{i} \ge 0$,  where $i=1,\dots,k$, $\,j=1,\dots, d^*$ ;
\item partitions $\Delta^{i}_j$, where $i=1,\dots,k$, $\,j=1,\dots, d^*$ and $|\Delta^i_j|=d$;
\item partitions $\Delta^{(1)},\dots,\Delta^{(\textsc{f})}$, where $|\Delta^{(i)}|=d$;

\end{itemize}

\vspace{2ex}

Denote by
$$T^{\textsc{e}}_{d,N}(b|N_1,\dots,N_k|L_1,\dots,L_k|
\Delta^{(1)},\dots,\Delta^{(\textsc{f})})=$$

\be\label{special-N-l^*}
\sum \frac{1}{\prod_{1\le i \le k\atop 1\ge j\ge d^*} c^i_j !} H^{\textsc{e},\textsc{f}+b+\sum_i|{\bf c}^i|}_{d,N}(\Gamma_1,\dots,\Gamma_b,
c^1_1,\dots,c^1_{d^*},
\dots,c^k_1,\dots,c^k_{d^*},
\Delta^{(1)},\dots,\Delta^{(\textsc{f})}),
\ee
where the sum is taken over partitions  $\Delta_i$ of the weight $d$ and over all partitions $\Delta^i_j$ of the same weight
that are not equal to $(1,\dots,1)$, such that
\begin{itemize}
\item $\Gamma_1=\dots=\dots,\Gamma_b=[2,1\dots,1]$;
\item  $\sum\limits_{j=1}^{d^*}\, c^i_j\,=\,N_i$, \quad   i=1,\dots,k
\item $\sum\limits_{j=1}^{d^*}\, c^i_j\ell(\Delta^i_j))=L_i$, \quad   i=1,\dots,k\,.
\end{itemize}

We recall that for $N\ge d$ the numbers $H^{\textsc{e},p}_{d,N}$ are the Hurwitz numbers of the $d-$fold
coverings with $p$ branch points of the surface with Euler characteristic $\textsc{e}$. 

\vspace{2ex}

Using the Example 4 in Subsection \ref{Weighted-sums-Section}
we obtain the following corollary 
of the Theorem \ref{Th-BKP} we get
\bc
$$
(-1)^k\sum_{n_1,\dots,n_k < 0}\,\prod_{i=1}^k\frac{(1-z_i)^{-n_i}}{n_i}\,\tau^{1,1}
\left(N,n,\bpow,\,\beta\,|\,\{n_i,a_i \}\,\right)=
$$

\[
\sum_{\Delta  } \,e^{\beta|\lambda|n}\,\bpow_{\Delta} \,
\frac{\beta^b}{b!}\,\times
\]
\be\label{T-BKP}
T^{(1)}_{d,N}(b|N_1,\dots,N_k|L_1,\dots,L_k|
\Delta^{(1)},\dots,\Delta^{(2)})\,\prod_{i=1}^k\,\left(z_i^{-N_i}-1\right)\,
(-1)^{N_i}
\Gamma(N_i)\,(a_i+n)^{L_i}
\ee
where $z_i$ and $a_i$ , $i=1,\dots,k$ are free parameters and where
\[
 \tau^{1,1}
\left(N,n,\bpow,\,\beta\,|\,\{n_i,a_i \}\,\right)=
\sum_{\lambda\atop\ell(\lambda)\le N}\,e^{\beta f_2(\lambda,n)}\,\prod_{i=1}^k\,
\left((a_i+n)_\lambda \right)^{n_i}\,s_\lambda(\bpow)
\]
is the BKP  tau function.

\ec

By the Theorem \ref{Th-KP}
\bc
$$
(-1)^k\sum_{n_1,\dots,n_k < 0}\,\prod_{i=1}^k\frac{(1-z_i)^{-n_i}}{n_i}\,\tau^{2,2}
\left(N,n,\bpow^{(1)},\bpow^{(2)}\,\beta\,|\,\{n_i,a_i \}\,\right)=
$$

\[
\sum_{\Delta^{(1)},\Delta^{(2)}  } \,e^{\beta|\lambda|n}\,\bpow^{(1)}_{\Delta^{(1)}}\bpow^{(2)}_{\Delta^{(2)}} \,
\frac{\beta^b}{b!}\,\times
\]
\be\label{T-KP}
T^{\textsc{e}}_{d,N}(b|N_1,\dots,N_k|L_1,\dots,L_k|
\Delta^{(1)},\Delta^{(2)})\,\prod_{i=1}^k\,\left(z_i^{-N_i}-1\right)\,
(-1)^{N_i}
\Gamma(N_i)\,(a_i+n)^{L_i}
\ee
where $z_i$ and $a_i$ , $i=1,\dots,k$ are free parameters and where
\[
 \tau^{2,2}
\left(N,n,\bpow^{(1)},\bpow^{(2)}\,\beta\,|\,\{n_i,a_i \}\,\right)=
\sum_{\lambda\atop\ell(\lambda)\le N}\,e^{\beta f_2(\lambda,n)}\,\prod_{i=1}^k\,
\left((a_i+n)_\lambda \right)^{n_i}\,s_\lambda(\bpow^{(1)})s_\lambda(\bpow^{(2)})
\]
is the 2KP (or, the same, of the semiinfinite TL) tau function.

\ec

\paragraph{Alternating sums.}

\vspace{2ex}

Here we present other combination of Hurwitz numbers weighted by $\pm 1$.

We consider the sums of Hurwitz numbers which correspond to the following options:
\begin{itemize}
\item $b,d,k,s\in\mathbb{Z}$, where $d>0$, $b,k,s\geq0$;
\item $l_1,\dots,l_k\in\mathbb{Z}$, where $0\le l_i\le d$;
\item $l_1^*,\dots,l^*_s\in\mathbb{Z}$, where $0< l^*_i \le d$ (the first uniquality is strict);
\item partitions $\Delta^{(1)},\dots,\Delta^{(\textsc{f})}$, where $|\Delta^{(i)}|=d$;
\item partitions $\Delta_1,\dots,\Delta_k$, where $|\Delta_i|=d$;
\item $s$ sets of partitions $\Delta_1^{i},\dots,\Delta_{m_i}^{i}$, $i=1,\dots,s$, where $|\Delta_j^i|=d$ for each $i,j$.
\end{itemize}

\vspace{2ex}

Denote by
$$S^{\textsc{e}}_{d,N}(b|l_1,\dots,l_k|l_1^*,\dots,l^*_s|
\Delta^{(1)},\dots,\Delta^{(\textsc{f})})=$$

\be\label{special'''''}
\sum (-1)^{\varepsilon} H^{\textsc{e}}_{d,N}(\Gamma_1,\dots,\Gamma_b,
\Delta_1,\dots,\Delta_k,\Delta^1_1,\dots,\Delta^1_{m_1},
\dots,\Delta^s_1,\dots,\Delta^s_{m_s},
\Delta^{(1)},\dots,\Delta^{(\textsc{f})}),
\ee
where the sum is taken over partitions  $\Delta_i$ of the weight $d$ and over all partitions $\Delta^i_j$ of the same weight
that are not equal to $(1,\dots,1)$, such that
\begin{itemize}
\item $\Gamma_1=\dots=\dots,\Gamma_b=[2,1\dots,1]$;
\item $d-\ell(\Delta_i)=l_i$, \quad i=1,\dots,k;
\item $\sum\limits_{j=1}^{m_i}(d-\ell(\Delta^i_j))=l^*_i$, \quad   i=1,\dots,s;
\item $\varepsilon=m_1+\dots+m_s$.
\end{itemize}

Thus, the numbers
$S^{\textsc{e}}_{d,N}(b|1,\dots,1|1,\dots,1|\Delta^{(1)},\Delta^{(2)})$
are  the 2-Hurwitz numbers from \cite{Okounkov-2000} with profiles $\Delta^{(1)}$ and $\Delta^{(2)}$ in
two given points with additional $b+k+s$ simple ramification points.

\vspace{2ex}

It follows from (\ref{Hurwitz}) that $$S^{\textsc{e}}_{d,N}(b|l_1,\dots,l_k|l_1^*,\dots,l^*_s|
\Delta^{(1)},\dots,\Delta^{(\textsc{f})})$$ we can present as the sum, where summands correspond to Young diagram
$\lambda$ of weight $d$.

\paragraph{Hurwitz weighted sums (\ref{special'''''}) for the BKP hierarchy.} 
Let us write down the following particular case of the 
Theorem \ref{Th-BKP}. 
Construct now  generating function for $S^N_{\mathbb{RP}^2}(d|l_1,\dots,l_k|l_1^*,\dots,l^*_s|
\Delta)$. Put
\be\label{BKP-generating-Hurwitz-sums} \tau^{BKP}(N,n|a_1,\dots,a_k|b_1,\dots,b_s|{\bpow}_{\Delta})=
\sum\limits_B \,(qe^{\beta n})^d\frac{\beta^b}{b!}  \prod_{i=1}^s
\,b_i^{-d}\,(b_i+n)^{-l^*_i} \prod_{i=1}^k a_i^{-d}(a_i+n)^{-l_i}\,
S^N_{\mathbb{RP}^2}(B)\, p_{\Delta}
\ee
where the sum is taken over all $B=(d|l_1,\dots,l_k|l_1^*,\dots,l^*_s|\Delta)$. Here $p_{\Delta}=p_{d_1}p_{d_2}\cdots$
with $\Delta=(d_1,d_2,\dots)$ and $\Delta=(d_1,d_2,\dots)$, $d_1+d_2+\dots =d$.

\bc \label{corollary-BKP}
The function $\tau^{BKP}(N,n|a_1,\dots,a_k|b_1,\dots,b_s|\bpow)$ is a $\tau-$ functions for
BKP hierarchy for any complex numbers $(a_1,\dots,a_k|b_1,\dots,b_s)$ and integer $n>0$.
\ec

\paragraph{Hurwitz weighted sums (\ref{special'''''}) for the 2KP (the same, for the semiinfinte TL) hierarchy.}

Consider
\bea\label{TL-generating-Hurwitz-sums} \tau^{TL}(M,n|q,\beta|a_1,\dots,a_k|b_1,\dots,b_s|
{\bpow}, {\bbpow})=\\
\sum\limits_B\,(qe^{\beta n})^d\frac{\beta^b}{b!}
\prod_{i=1}^k \,a_i^{-d}(a_i+n)^{-l_i} \prod_{i=1}^s \,b_i^{-d}\,(b_i+n)^{-l^*_i}
S_{\mathbb{CP}^1}^{n-M}(B|\Delta^{(1)},\Delta^{(2)})\, p_{\Delta^{(1)}} {\bar p}_{\Delta^{(2)}}
\eea
where $n>M$ and the sum is taken over all $B=(d|l_1,\dots,l_k|l_1^*,\dots,l^*_s|\Delta^{(1)},\Delta^{(2)})$.
Here $p_{\Delta^{(i)}}=p_{d_1^{(i)}}p_{d_2^{(i)}}\cdots$
with $\Delta^{(i)}=(d_1^{(i)},d_2^{(i)},\dots)$, $d_1^{(i)}+d_2^{(i)}+\dots =d$,  $i=1,2$.

Here we basically review the result of \cite{HO-2014} with certain modifications, namely, we need not infinite
but semiinfinite Toda lattice to compare with our main result in the next section. We do not consider the combinatorial
interpretation of hypergeometric tau functions in terms of counting paths problem in Cayley graph related to the symmetric group
worked out in \cite{Harnad-2014}, \cite{HO-2014}.

\bc The function $\tau^{TL}(M,n|a_1,\dots,a_k|b_1,\dots,b_s| \bpow,
\bbpow)$ is a $\tau$ functions for semiinfinte 2DToda hierarchy for any complex numbers
$(a_1,\dots,a_k|b_1,\dots,b_s)$.
\ec
It directly follows from the Theorem \ref{Th-KP}.

\br
The special case  $s=0$ was pointed out in \cite{AMMN-2014}.
\er

\section{Transformation of Hurwitz $\tau$-functions for the semiinfinte 2DToda hierarchy to $\tau$-functions for BKP hierarchy
\label{transformations-section}}

The sum (\ref{BKP-sums}) is an example of the BKP hypergeometric tau function \cite{OST-I} (Hirota equations
for the BKP tau functions may be found in Appendix \ref{N-n-BKP-Hirota}).
It may be obtained from
the  hypergeometric tau function of the semi-infinite Toda lattice $n\ge M$ (here a given $M$ is the origin of the lattice):
\[
 \frac{\partial\phi_n}{\partial p_1\partial {\bar p}_1}=
 { r}'(n)e^{\phi_{n-1}-\phi_n}-{ r}'(n+1)e^{\phi_n-\phi_{n+1}}\,,\quad
 e^{-\phi_n}=\frac{\tau^{}_{ r'}(M;n+1,\bpow,\bbpow)}{\tau_{ r'}(M;n,\bpow,\bbpow)}\frac{g(n)}{g(n+1)}
\]
where ${ r}'(n)=r(n)\delta(n)$ ($\delta(M)=0$, $\delta(n)=1$ otherwise). The multiplication by $\delta$ provides
the restriction of the summation region by the condition $r'_\lambda(n)$ for $\ell(\lambda)\le n-M$ for the tau function
of the Toda lattice, this we see from the definition of $r_\lambda(n)$.
Fixing $n$ and choosing $N=n-M$ we obtain
\be\label{relation}
\tau_r^{\rm BKP}(N,n,\bpow) \,=\, \left[
e^{\frac 12 L_\infty} \cdot \tau_{r'}^{\rm TL}(n-N;n,\bpow,\bbpow)
 \right]_{\bbpow =0}
\ee
where $L_\infty$ is the following Laplacian operator
\be
L_\infty \,=\,
\sum_{m\ge 1}\, m \frac {\partial^2}{\partial {\bar p}_m^2}\,+
\,2\sum_{m\ge 1,{\rm odd}} \,m \frac{\partial}{\partial {\bar p}_{m}}
\ee

The operator $e^{\frac 12 L_\infty}$ and the evaluation at $\bbpow=0$ 'eliminate' one
Schur function in each term of (\ref{TL-sums}).

The proof of (\ref{relation}) follows from
 \[
\sum_{\lambda\in\Pa}\,s_\lambda(\bpow)\,=\,
e^{\frac 12\sum_{m=1}^\infty\,\frac 1m p_m^2\,+\,\sum_{m=1}^\infty \,  p_{2m-1}}
 \]
and
\[
 \left[s_\mu({\tilde\partial})\cdot s_\lambda(\bpow) \right]_{\bpow =0}\,=\,\delta_{\mu,\lambda}
\]
which may be derived from Examples in Chapter I section 5  of \cite{Mac}. Here $s_\lambda(\tilde\partial)$
denotes the Schur function as the function defined by (\ref{Schur-pol}) where each $p_m$ is replaced by
$m\frac{\partial}{\partial p_m}$.

\section{Matrix integrals as generating functions of Hurwitz numbers and of Hurwitz generating series
\label{Matrix-integrals}}

We want to present Hurwitz generating series (\ref{hyp-tau-e-f}) in form of matrix integrals. 
The idea is that various integrals of products of the 2KP and BKP tau functions of matrix argument
\footnote{
If $\bpow=(p_1,p_2,\dots)$ and $p_m=\tr M^m$ where $M$ is a matrix, we may write $\tau=\tau(M):=\tau\left(\bpow(M) \right)$ and call it
tau function of matrix argument} result in the series (\ref{hyp-tau-e-f}). 

\paragraph{One-matrix model and Hurwitz numbers.}
To begin with we recall that the partition function of such standard matrix models as the two-matrix model, the model
of complex matrices (complex Ginibre ensemble) and model of normal matrices may be written in form of the following
perturbation series in coupling constants $\bpow^{(1)},\bpow^{(2)}$
\be\label{MM-Schur-series}
 \sum_{\lambda\atop \ell(\lambda)\le N}\,\frac{s_\lambda(I_N)}{s_\lambda(\bpow_\infty)}\,
 s_\lambda(\bpow^{(1)})s_\lambda(\bpow^{(2)})
\ee
see \cite{HO-2MM}, \cite{O-2003}, \cite{OShiota-2004}. Taking into account explanations above, and
\[
 s_\lambda(I_N)\,=\, \chi_\lambda(1)N^d\,
 \left(1+\sum_{\Delta\neq 1^d } |C_{\Delta}|\frac{\chi_\lambda(\Delta)}{\chi_\lambda(1)} N^{\ell(\Delta)-d} \right)
\]
This series is of form (\ref{hyp-tau-e-f}) where $k=1$ and where $q_1=t^N_1\to 1$ which, as we know, 
generates sums of Hurwitz numbers for the sphere  with three ramification points where two profiles are arbitrary and fixed
and summation is ranged over all profiles in the third point whose length is a given  number. At first this fact was marked
in \cite{AMMN-2014}. 

Our remark to this observation is as follows. 
Series (\ref{MM-Schur-series}) also describes the celebrated one-matrix model in case 
$p^{(2)}_2\neq 0,\, p^{(2)}_i=0,\,i>2$, for the Schur function seires see \cite{HO-scalar}. 
Thus we obtain

\begin{itemize}

\item
The one-matrix model generates sums of all Hurwitz numbers for the sphere 
under the following conditions: the profile at $\infty$ is a given partition, the profile at $0$ 
is a given partition of type $1^{d_1}2^{d_2}$,  the length of the profile at the third point is a given number.

\end{itemize}

\br
To compare results,
one should have in mind that in matrix model studies they typically send $p_i \to N p_i,\, {\bar p}_i \to N {\bar p}_i$ 
in the $N\to\infty$ limit which yields the following behavior in $N$:
\[
 s_\lambda(\bpow)\,=\, \chi_\lambda(1)N^d\,
 \left(p_1^d+\sum_{\Delta\neq 1^d } |C_{\Delta}|
 \frac{\chi_\lambda(\Delta)}{\chi_\lambda(1)}N^{\ell(\Delta)-d}\bpow_{\Delta} \right)
\]
\er

{\bf Example}.
The partition function for the Hermitian  one-matrix model which is related to the 'triangulation'
by $k$-gons and describes a model of two-dimensional gravity is
\[
Z(N,g,g_k)=\int dM e^{-N\tr(\frac g2 M^2+\frac{g_k}{k} M^k)} = 
\frac{s_\lambda(I_N)}{s_\lambda(\bpow_\infty)}\,
s_\lambda\left(0,p_2=-\frac{1}{2Ng},0,\dots\right)s_\lambda\left(0,\dots,0,p_k=-Ng_k,0,\dots\right)
\]
\be
=\,\sum_{F,L,V}\,Z_{F,L,V} \,N^{F-L+V}\,g^{-L}\,g_k^V
\ee
where $F,L,V$ are the number of faces, lines and vertices of the Feynman fat graph which contributes to $Z_{F,L,V}$.
\footnote{ According to the Feynman rules for the one matrix
model we have: (a) to each propagator (double line) is associated a
factor $1/(Ng)$ (which is $p_2^{(1)}$ in our notations) (b) to each
four legged vertex is associated a factor $(-Ng_4)$ (which is $Np_4^{(2)}$
in our notations) (c) to each closed single line is associated a
factor $N$. Therefore we may say that the factor $\frac{s_\lambda(I_N)}{s_\lambda(\bpow_\infty)}$
in (\ref{MM-Schur-series}) is responsible for closed lines, the factor
$s_{\lambda}({\bf t}=0,Np_2^{(1)},0,\dots)$ is responsible for
propagators and the factor $s_{\lambda}(N\bpow^{(2)})$ is
responsible for vertices.} 
(We choose the normalization of the matrix integral in such a way
that $Z(N,g,g_k)$ is equal to $1$ when $g_k=0$.)
Thus we obtain
\be
Z_{F;L;V}\,=\,\sum_{\Delta \atop \ell(\Delta)=F}\,H_{\mathbb{CP}^1} \left(\Delta, (2^L), (k^V)\right)
\ee
where $|\Delta|=2L=kV$, otherwise both sides vanish.
After taking the logariphm we abotain the same for the connected Feynman graphs and the connected Hurwitz numbers:
\be
Z_{F;L;V}^{\rm connected}\,=\,\sum_{\Delta \atop \ell(\Delta)=F}\,H^{\rm connected}_{\mathbb{CP}^1} 
\left(\Delta, (2^L), (k^V)\right)
\ee

In case there is a set of coupling constants $g_3,g_4,\dots$ instead of a single one, $g_k$, we obtain 
\[
 Z(N,g,g_3,\dots)=  \sum_{F,L,V}\,Z_{F,L,V} \,N^{F-L+\sum_{k>2} V_k}\,g^{-L}\,\prod_{k>2}\,g_k^{V_k}
\]
where $Z_{F;L;V_3,V_4,V_5,\dots}$ is equal to the sum of the 3-point Hurwitz numbers,
\be
Z_{F;L;V_3,V_4,V_5,\dots}\,=\,\sum_{\Delta \atop \ell(\Delta)=F}\,
H_{\mathbb{CP}^1} \left(\Delta, \Delta_1, \Delta_2\right)\,,\qquad
\Delta_1=(2^L)\,,\,\,\Delta_2=\left( 3^{V_3}4^{V_4} 5^{V_5}\cdots \right)
\ee
and the Euler characteristic of related fat graphs is equal to the alternating sum of the lengths
of profiles: $\textsc{e}_{\rm fat\,graph}=\ell(\Delta)-\ell(\Delta_1)+\ell(\Delta_2)$.
We also obtain $|\Delta|=|\Delta_1|=|\Delta_2|$.

In \cite{AMMN-2014} some other interesting examples of relations between matrix models and sums of Hurwitz 
numbers were mentioned. For example it was presented a $2n$-matrix model with expansion (\ref{MM-Schur-series}) where
$\frac{s_\lambda(I_N)}{s_\lambda(\bpow_\infty)}$ was replaced by the $n$-th power of this factor.

Below we will write down matrix integrals which generate Hurwitz numbers themselves rather than weighted sums
(as (\ref{MM-Schur-series}) does),  though we also present some generating series for the sums.

\paragraph{Notation. Useful relations.}

We shall exploit the known formulae for integrations of Schur functions over the
unitary group and over complex matrices. Earlier they were used in \cite{O-2003} to clarify some links between
matrix models and integrable hierarchies. 

As in \cite{O-2003} instead of power sums written by small bold characters (like $\bpow$) sometimes where it is suitable 
we shall use the matrix arguments written by large character. Say, $s_\lambda(X)$ and $\tau(X)$ are respectively equal 
to  $s_\lambda(X):=s_\lambda\left(\bpow(X)\right)$ and  $\tau(X):=\tau \left(\bpow(X)\right)$ where
 $\bpow(X)=(p_1(X),p_2(X),\dots)$, where $p_m(X)=\tr X^m$.

 We use the following notations
 \begin{itemize}
  \item $  d_*U $ is the normalized Haar measure on $\mathbb{\mathbb{U}}(N)$: $\int_{\mathbb{U}(N)}d_*U =1$
  
  \item $Z$ is a complex matrix, $Z=UX(1+J)U^\dag$ (the Schur decomposition), where $X=\diag (z_i)$ is diagonal, 
  $J$ is strictly upper triangle, $U\in\mathbb{\mathbb{U}}(N)$
    $$
d\Omega^\texttt{C}(Z,Z^\dag)  =\,\pi^{-n^2}\,e^{-\tr \left(ZZ^\dag\right)}\,
\prod_{i,j=1}^N \,d \Re Z_{ij}d \Im Z_{ij} 
  $$
  \[
 = \,c_\texttt{Z} \,
d_*U\,||\,\prod_{N\ge i>j}\,|z_i-z_j|^2\,
\prod_{i=1}^N\,e^{-|z_i|^2}\,d^2z_i \,\left[e^{-\tr JJ^\dag} d^2 J_{ij}\right]  
  \]
  where the part related to the upper triangular factor in brackets is not important for our problems.
  
  \item $M$ is a normal matrix, $Z=UXU^\dag$, where $X=\diag (z_i)$ is diagonal,  
  $U\in\mathbb{\mathbb{U}}(N)$  
  $$  
  d\Omega^\texttt{N}(M,M^\dag)\,= \,\pi^{-n^2}\,e^{-\tr\left(MM^\dag \right)}\,
\prod_{i,j=1}^N d \Re M_{ij}d \Im M_{ij}  
  $$
  \[
   =\,c_\texttt{M}\, d_*U \,
\prod_{N\ge i>j}|z_i-z_j|^2\,\prod_{i=1}^N \,e^{-|z_i|^2}
\,d^2 z_i
  \]

  \item $H^{(1)}$ is a Hermitian matrix and $H^{(2)}$ is anti-Hermitian one, 
  $H^{(c)}=U^{(c)}X^{(c)}U^{(c)\dag}$, $X^{(c)}=\diag(x_i^{(c)})$,
  $U,U^{(c)}\in\mathbb{U}(N)$, $c=1,2$. Measure 
   $$
   d\Omega^\texttt{H}(H^{(1)},H^{(2)})= \,\int_{\mathbb{U}(N)}\,
e^{-\tr \left(H^{(1)} UH^{(2)}U^\dag\right)}d_*U\, \prod_{i\le j} 
d\Re H^{(1)} d\Im H^{(2)}\prod_{i<j} d\Im H^{(1)} d\Re H^{(2)}
  $$ 
  \[
  =\,c_\texttt{H}\, \prod_{c=1,2} d_*U^{(c)}
\prod_{N\ge i>j}(x^{(c)}_i-x^{(c)}_j)\prod_{i=1}^N e^{-x^{(1)}_i x^{(2)}_i}d x^{(1)}_i d  x^{(2)}_i
  \]
 \end{itemize}
where the constants $c_a$, $a=\texttt{C,N,H}$,  are chosen for normalization: $\int d\Omega_\rho^{(a)}=1$.
\br \label{notation-M-M^*} In what follows,
for unification and to save space,
 we shall use the notation $M$ and $M^*$ replacing the pairs $Z,Z^\dag$, $M,M^\dag$ and also $H^{(1)},H^{(2)}$.
In the last case the matrices $M$ and $M^*$ are not related by the Hermitian conjugation.
\er

These measures provides the relation
\be\label{s-s-N_lambda-1}
\int s_\lambda(M)s_\mu(M^*)\,d\Omega^a(M,M^*) = (N)_\lambda\delta_{\lambda,\mu}
\ee
where $a=\texttt{C},\texttt{N},\texttt{H}$. This relation was used in 
\cite{O-Acta},\cite{HO-2MM},\cite{O-2003},\cite{AMMN-2014},\cite{OShiota-2004}, for for models of Hermitian, complex 
and normal matrices.
\footnote{If we 
replace the factor $e^{\tr\left(MM^* \right)}$ in the measure 
$d\Omega^a$ by a hypergeometric tau function $\tau_r(N,MM^*,I_N)$, then the factor $(N)_\lambda$ in the right 
hand side of (\ref{s-s-N_lambda}) should be replaced by $\frac{1}{r_\lambda(N)}$ \cite{O-Acta}.}

By $I_N$ we shall denote the $N\times N$ unit matrix.
Then (for instance see \cite{Mac})
\bl \label{useful-relations'}
Let $A$ and $B$ be normal  matrices (i.e. matrices diagonalizable by unitary transformations). Then
\begin{equation}\label{sAUBU^+1}
\int_{\mathbb{U}(N)}s_\lambda(AUBU^{-1})d_*U=
\frac{s_\lambda(A)s_\lambda(B)}{s_\lambda(I_n)} \ ,
\end{equation}
For $A,B\in GL(N)$ we have
\begin{equation}\label{sAUU^+B'}
\int_{\mathbb{U}(n)}s_\mu(AU)s_\lambda(U^{-1}B)d_*U=
\frac{s_\lambda(AB)}{s_\lambda(I_N)}\delta_{\mu,\lambda}\,.
\end{equation}
Below ${\bf p}_{\infty}=(1,0,0,\dots)$. 
\begin{equation}\label{sAZBZ^+'}
\int_{\mathbb{C}^{n^2}} s_\lambda(AZBZ^+)e^{-\textrm{Tr}
ZZ^+}\prod_{i,j=1}^n d^2Z=
\frac{s_\lambda(A)s_\lambda(B)}{s_\lambda({\bf p}_{\infty})}
\end{equation}
and
\begin{equation}\label{sAZZ^+B'}
\int_{\mathbb{C}^{n^2}} s_\mu(AZ)s_\lambda(Z^+B) e^{-\textrm{Tr}
ZZ^+}\prod_{i,j=1}^nd^2Z= \frac{s_\lambda(AB)}{s_\lambda({\bf
p}_{\infty})}\delta_{\mu,\lambda}\,.
\end{equation}
\el
We recall that
$$ s_\lambda(I_N)=(N)_\lambda s_\lambda(\bpow_\infty)\,,\qquad\chi_\lambda(1)=|\lambda|!s_\lambda(\bpow_\infty) $$

Lemma \ref{useful-relations'} allows to pick up Hurwitz numbers from matrix integrals in many ways. 
Details may be found in the Appendix \ref{Matrix-integrals-appendix}
Below we describe the simplest examples of integrals of tau function which are not tau function.

We use the simplest (the so-called vacuum) 2-KP tau function 
\be\label{vac-tau-2KP'}
\tau_1^{\rm 2KP}(X,\bpow)\,:=\,\sum_\lambda \,s_\lambda(X)\,s_\lambda(\bpow)=e^{\tr V(X,\bpow)},\quad
\tau_1^{\rm 2KP}(X,\bpow_\infty)\,=e^{\tr X}
\ee
where
\[
 V(z,\bpow)=\sum_{m>0}\,
 \frac{z^mp_m}{m}
\]
$z$ may be a number and may be a matrix. 

We use the simplest BKP tau function 
\cite{OST-I}
\be\label{vac-tau-BKP'}
 \tau_1^{\rm BKP}(X)\,:=\,\sum_\lambda \,s_\lambda(X)\,=\,\prod_{N>i>j}\,(1-x_ix_j)^{-1}\,\prod_{i=1}^N\,(1-x_i)^{-1}
\ee
Then we have

{\bf Example}. $\mathbb{CP}^1$ with three ramification point. In particular we have the following integrals
generating Hurwitz numbers
\[  
\sum_{\lambda\atop\ell(\lambda)\le N} \frac{s_\lambda(\bpow^{(1)})
s_\lambda(\bpow^{(2)})s_\lambda(\bpow^{(3)})}{s_\lambda(\bpow_\infty)}\, =
\]
\[
= \,\int e^{\tr V(M_1 M_2,\,\bpow^{(3)})}\prod_{i=1,2} e^{\tr V(M^*_i,\,\bpow^{(i)})} d\Omega^a(M_i,M^*_i)
\]
\[
=\,\int e^{\tr V(\Lambda M,\,\bpow^{(1)})+ \tr V(M^*,\,\bpow^{(2)})} \,d\Omega^a(M,M^*)\,,\quad p_m^{(3)}=\tr \Lambda^m
\]
where $a=\texttt{C,N,H}$.

{\bf Example}. An analog of (\ref{MM-Schur-series}) for $\mathbb{RP}^2$ with three ramfication points with two 
arbitrary profiles at $0$ and at $\infty$ with fixed length in the third point:
\[
 \sum_{\lambda}\frac{s_\lambda(I_N)s_\lambda(\bpow^{(1)})s_\lambda(\bpow^{(2)})}{\left(s_\lambda(\bpow_\infty) \right)^2}
\]
\[
 = \,\int \,\tau^{\rm BKP}_1\left( M_1 M_2 \right)  \,\prod_{i=1,2} \,
  e^{V(\tr M^*_i,\,\bpow^{(i)})}\,d\Omega^a(M_i,M^*_i)\,\quad a=\texttt{C,N,H}
\]
\[
=\, \int \, e^{\tr\left(\Lambda M_1 M_2 \right) }\,\tau^{\rm BKP}_1(M_1^*)\,e^{\tr V(M^*_2,\,\bpow)} \,
\prod_{i=1,2}\,d\Omega^{\texttt{C}}(M_i,M^*_i)\,, \qquad p^{(2)}_m =\tr \Lambda^m
\]

\section*{Acknowledgements}

A.O. was supported by RFBR grant 14-01-00860 and by V.E.Zakharov's scientific school (Leading schientific schools). He thanks
John Harnad, Leonid Chekhov and Andrei Mironov for explanations related to their works and John Harnad, Anton Zabrodin and 
Johan van de Leur for fruitful discussions. We thanks Sergei Loktev for a useful remark.
Our special grates to Chekhov for the organization of the workshop on Hurwitz numbers (Moscow, May 2014) which inspired us to 
do this work.
The work of S.N.
was partially supported by Laboratory of Quantum Topology of Chelyabinsk State University
 (Russian Federation government grant 14.Z50.31.0020), by RFBR grants 13-02-00457  and NSh-5138.2014.1.

\appendix

\section{Appendices}

\subsection{Macdonald polynomials \label{MacdPol}}

This Appendix is a result of a discussion with John Harnad.

One can write down the scalar product where Macdonald polynomials $P_\lambda(q^a,q;\bpow)$ are orthonormal, by the integral
over all power sums variables $\bpow$ as follows
\be
\l f,g \r = \int_{\mathbb{C}^\infty}\,f(\bpow)g(\bpow^*)\,\prod_{m=1}^\infty\, e^{-\frac 1m |p_m|^2 p_m(a,q)}\,
\frac{p_m(a,q)}{2\pi i m }\,{dp_m\wedge dp_m^*}
\ee
where $p_m^*$ is the complex conjugate to $p_m^*$.  In this basis
\[
 \l p_\Delta p_{\Delta'} \r = \frac{1}{w_\Delta(a,q)}\delta_{\Delta,\Delta'}\,,\quad
 w_\Delta(a,q)=\frac{p_\Delta(a,q)}{z_\Delta}
\]
Exactly this ratio appears in the character expansion formula
\be
s_\lambda(\bpow(a,q)) = \sum_\Delta \chi_\lambda(\Delta) w_\Delta(a,q)
\ee
Also one can write
\be
\sum_{\Delta\in\Pa} \,\frac{p_\Delta(a,q)}{z_\Delta}P_\Delta(\bpow)P_\Delta(\bbpow)=
e^{\sum_{m>0}\frac 1m p_m {\bar p}_m p_m(a,q)}
=\sum_{\Delta\in\Pa}\frac{p_\Delta(a,q)}{z_\Delta} \sum_{\lambda\in\Pa} s_\lambda(\bpow\bbpow)\chi_\lambda(\Delta)
\ee
where $\bpow\bbpow$ denotes the set $(p_1{\bar p}_1,p_2{\bar p}_2,\dots)$.

Equation (\ref{bpow(a)}) corresponds to the case $q\to 1$ where Macdonald polynomials convert to Jack ones.

The relation of of hypergeometric tau functions to the quantum integrable systems is unclear.
The combinatorial and geometric interpretations of hypergeometric tau functions parametrized by pairs $a_i,q_i$ in the
TL case will be considered in \cite{HO-future}.

 \subsection{Hirota equations for the BKP tau function with two discrete time variables.\label{N-n-BKP-Hirota}}

 The BKP hierarchy we are interested in was introduced in \cite{KvdLbispec}. It was used
 to construct various matrix models \cite{L1}, \cite{OST-I}, \cite{O-2012}.
 Hirota equations for the BKP hierarchy of Kac-van de Leur were presented in \cite{KvdLbispec}.
 However in our case we need more general version which includes both discrete variables
 $N$ and $n$, see \cite{OST-I}.
 The BKP tau function we need has the following form
 \be\label{N-n-BKP}
 \tau^{\rm BKP}(N,n,\bpow |g)=\l N+n|e^{\sum_{m>0} \frac 1m {\bar p}_m J_m} g |n\r
 \ee
 where Clifford algebra element $g$ may be considered as an element of $\mathbb{O}(2\infty +1)$ group which
 specifies the choice of the BKP tau function,
 \[
  J_m=\sum_{i\in\mathbb{Z}}\,:\psi_i\psi^\dag_{i+m}:
 \]
  are Fourier modes of current operators, see details in \cite{OST-I}. Hirota equations for tau function (\ref{N-n-BKP}) may be
  obtained by a certain specification of the Hirota equations for the two-sided BKP tau function
  \[
 \tau^{\rm BKP}(N,n,\bpow,\bbpow |g)=\l N+n|e^{\sum_{m>0} \frac 1m p_m J_m} g e^{\sum_{m>0} \frac 1m p_m J_{-m}}|n\r
  \]
see \cite{OST-I},  which in our notations are
 \bea\label{Hirota-two-sided-BKP}
  \oint\frac{dz}{2\pi i}z^{N'+n'-N-n-2}e^{V(\bpow'-\bpow,z)}
  \tau(N'-1,n',\bpow'-[z^{-1}],{\bbpow}')
  \tau(N+1,n,\bpow+[z^{-1}],{\bar \bpow}) \nonumber\\
+ \oint\frac{dz}{2\pi i}z^{N+n-N'-n'-2}e^{V(\bpow-\bpow',z)}
  \tau(N'+1,n',\bpow'+[z^{-1}],{\bar \bpow}')
  \tau(N-1,n,\bpow-[z^{-1}],{\bar \bpow}) \nonumber\\
= \oint\frac{dz}{2\pi i}z^{n'-n}e^{V({\bar \bpow}'-{\bar \bpow},z^{-1})}
  \tau(N'-1,n'+1,\bpow',{\bar \bpow}'-[z])
  \tau(N+1,n-1,\bpow,{\bar \bpow}-[z]) \nonumber \\
+ \oint\frac{dz}{2\pi i}z^{n-n'}e^{V({\bar \bpow}'-{\bar \bpow},z^{-1})}
  \tau(N'+1,n'-1,\bpow',{\bar \bpow}'+[z])
  \tau(N-1,n+1,\bpow,{\bar \bpow}+[z]) \nonumber\\
+ \frac{(-1)^{n'+n}}{2}(1-(-1)^{N'+N})
  \tau(N',n',\bpow',{\bar \bpow}')\tau(N,n,\bpow,{\bar \bpow})
\eea
see also \cite{LeurO-2014}. Here $\bpow=(p_1,p_2,\dots)$, $\bpow'=(p'_1,p'_2,\dots)$, $\bbpow=({\bar p}_1,{\bar p}_2,\dots)$,
$\bbpow'=({\bar p}'_1,{\bar p}'_2,\dots)$, and
\[
 V(z,\bpow)=\sum_{m>0} \frac 1m z^m p_m
\]
The notation $\bpow +[z^{-1}]$ denotes the set $\left( p_1+z^{-1},p_2+z^{-2}, p_3+z^{-3},\dots \right)$.

 \br
Actually up to some simple factor the two-sided BKP tau function of \cite{OST-I} coincides with the two-component
BKP tau function of \cite{KvdLbispec} and Hirota equations (\ref{Hirota-two-sided-BKP})
basically coincide with the Hirota equations for the two-component
BKP, see Appendix in \cite{LeurO-2014}.
 \er

 To obtain Hirota equations for (\ref{N-n-BKP}) we chose ${\bar \bpow}={\bar \bpow}'=0$.

 For $n'=n+1$, we obtain (see \cite{OST-I})
 \bea\label{Hirota-N-n-BKP-OST}
  \oint\frac{dz}{2\pi i}z^{N'-N-1}e^{V(\bpow'-\bpow,z)}
  \tau(N'-1,n+1,\bpow'-[z^{-1}]|g)
  \tau(N+1,n,\bpow+[z^{-1}]|g) \nonumber\\
+ \oint\frac{dz}{2\pi i}z^{N-N'-3}e^{V(\bpow-\bpow',z)}
  \tau(N'+1,n+1,\bpow'+[z^{-1}]|g)
  \tau(N-1,n,\bpow-[z^{-1}]|g) \nonumber\\
=
  \tau(N'+1,n,\bpow'|g)
  \tau(N-1,n+1,\bpow|g)
- \frac{1}{2}(1-(-1)^{N'+N})
  \tau(N',n+1,\bpow'|g)\tau(N,n,\bpow|g)
\eea
For $n'=n$, we obtain Hirota equations as in \cite{KvdLbispec}
\bea\label{Hirota-N-BKP(deLeur)}
  \oint\frac{dz}{2\pi i}z^{N'-N-2}e^{\xi(\bt'-\bt,z)}
  \tau(N'-1,n,\bt'-[z^{-1}])\tau(N+1,n,\bt+[z^{-1}]) \nonumber\\
+ \oint\frac{dz}{2\pi i}z^{N-N'-2}e^{\xi(\bt-\bt',z)}
  \tau(N'+1,n,\bt'+[z^{-1}])\tau(N-1,n,\bt'-[z^{-1}]) \nonumber\\
= \frac{1}{2}(1-(-1)^{N'+N})\tau(N',n,\bt')\tau(N,n,\bt)
\eea

Let us write down some of them. Taking $N'=N+1$ and all $p_i=p_i',\,i\neq 1$ in (\ref{Hirota-N-n-BKP-OST})
and picking up the terms linear in $p'_1-p_1$ we obtain
\bea
\frac 12 \tau(N,n+1,\bpow)\frac{\partial^2 \tau(N+1,n,\bpow)}{\partial^2 p_1}-
\frac 12 \frac{\tau(N,n+1,\bpow)}{\partial^2 p_1} \tau(N+1,n,\bpow)=\nonumber
\\
\frac{\partial\tau(N+2,n,\bpow)}{\partial p_1}\tau(N-1,n+1,\bpow)-
 \frac{\partial \tau(N+1,n+1,\bpow)}{\partial p_1}\tau(N,n,\bpow)
\eea
Taking $N'=N+1$ and all $p_i=p_i',\,i\neq 2$ in (\ref{Hirota-N-BKP(deLeur)})
and picking up the terms linear in $p'_2-p_2$ we obtain
\bea
\frac 12 \frac{\partial\tau(N,n,\bpow)}{\partial p_2} \tau(N+1,n,\bpow)-
\frac 12 \tau(N,n,\bpow)\frac{\partial\tau(N+1,n,\bpow)}{\partial p_2} \nonumber
+\frac 12 \frac{\partial^2\tau(N,n,\bpow)}{\partial^2 p_1} \tau(N+1,n,\bpow)\\
+\frac 12 \tau(N,n,\bpow)\frac{\partial^2\tau(N+1,n,\bpow)}{\partial^2 p_1}
- \frac{\partial\tau(N,n,\bpow)}{\partial p_1}\frac{\partial\tau(N+1,n,\bpow)}{\partial p_1}
=\tau(N+2,n,\bpow)\tau(N-1,n,\bpow)
\eea

\subsection{Fermionic formulae\label{fermionic-appendix}}
Details may be found in \cite{OS-2000, OST-I}.
Let
 $\{\psi_i$, $\psi_i^\dag$, $i \in \mathbb{Z}\}$ are Fermi creation and
annihilation operators that  satisfy the usual anticommutation relations and vacuum annihilation conditions
\be
[\psi_i^{(a)}, \ \psi^{\dag(b)}_j]_+ = \delta_{ij}\delta_{a,b}, \quad \psi^{(1)}_i | n,*\rangle =
\psi_{-i-1}^{\dag(1)} | n,* \rangle =0,\quad \psi^{(2)}_i | *,n\rangle =
\psi_{-i-1}^{\dag(2)} | *,n \rangle =0 \  \text{ if } \  i< n,
 \ee

 Sometimes we will omit the superscript $(1)$ in particular write  $\psi$ instead of $\psi^{(1)}$.

The hypergeometric tau functions may be written as follows
\[
 \tau^{\rm TL}_r(n,\bpow,\bbpow)=g(n)\l n| e^{\sum_{m>0} \frac 1m J_mp_m} e^{\sum_{m>0}\frac 1m p_m A_m}|n\r
\]
where $J_m=\sum_{i\in\mathbb{Z}}\psi_i\psi^\dag_{i+m}$ and
$A_m=\sum_{i\in\mathbb{Z}}r(i)\dots r(i-m)\psi_i\psi^\dag_{i-m}$.
The semiinfinite TL may be described either putting by $r(N)=0$ ,or, it is may be suitable to present it in form
\[
 \tau^{\rm TL}_r(M,n,\bpow,\bbpow)=(-1)^{\frac{M(M+1)}{2}}g(n)\l M+n,-M-n| e^{\sum_{m>0} \frac 1m J^{2}_mp_m -\frac 1m p_m A_m}
 e^{\sum_{n\in\mathbb{Z}} \psi_i^{(1)}\psi_{-i-1}^{\dag (2)}}\,|n,-n\r
\]
For BKP \cite{KvdLbispec} one needs to introduce an additional Fermi mode  $\phi$ which anticommutes with each other
Fermi operator except itself: $\phi^2=\frac 12$, and
 $\phi|0\r=\frac{1}{\sqrt{2}}|0\r$. Then
\be
 \tau^{\rm BKP}_r(N,n,\bpow,\bbpow)=g(n)\l N+n| e^{\sum_{m>0} \frac 1m B_m p_m} e^{\omega}|n\r =
 \l N+n| e^{\sum_{m>0} \frac 1m J_m p_m} e^{-\sum_{i\in\mathbb{Z}} U_i: \psi_i\psi_i^\dag :} e^{\omega}|n\r
\ee
and
\be
\tau^{\rm BKP}_r(N=\infty,0,\bpow)= \l 0| e^{\sum_{m>0} \frac 1m B_m p_m} e^{\omega}e^{\omega^\dag}|0\r= g(n)
\sum_{\lambda\in\Pa} r_\lambda(0)s_\lambda(\bpow)
\ee

where
\be\label{r-U}
r(i)=e^{U_{i-1}-U_{i}}
\ee
and
\[
\omega=\sum_{i>j} \psi_i\psi_j\,-\sqrt{2}\,\phi \sum_{i\in\mathbb{Z}} \psi_i
\]
\[
 \omega_-=\sum_{i>j\ge 0} \psi_i\psi_j\,-\sqrt{2}\,\phi\sum_{i\ge 0} \psi_i\,,
 \quad \omega_+ = \sum_{i>j\ge 0} (-)^{i+j}\psi^\dag_{-j-1}\psi^\dag_{-i-1}\,+\,\sqrt{2}\phi\sum_{i\ge 0} \psi^\dag_{-i-1},
\]

\[
 B_m=\sum_{i\in\mathbb{Z}}\frac{1}{r(i)}\dots \frac{1}{r(i+m)}\psi_i\psi^\dag_{i+m}
\]
and
\[
g(n)=\l n|e^{\sum_{i\in\mathbb{Z}} U_i: \psi_i\psi_i^\dag :}|n\r=
\]
\bea\label{g(n)}
 e^{-U_0+\cdots -U_{n-1}}\quad {\rm if}\,\, n>0 \\
1\quad {\rm if} \,\,n=0 \\
e^{U_{-1}+\cdots U_{n}}\quad {\rm if}\,\, n<0
\eea

\subsection{BKP tau functions.}

\paragraph{Hirota equations for multicomponent BKP.} 
 This is a particular case of the multicomponent BKP tau function, introduced in \cite{KvdLbispec},
 \be
 \tau({\bf N}; {\bf s}):=
 \l N^{(1)},\dots, N^{(p)} \vert e^{\sum_{a=1}^p\sum_{i>0} \beta^{a} s^{(a)}} h^{(1,\dots,p)} \vert 0,0\r\,
 \ee
where $h^{(1,\dots,p)}$ solves
\be\label{fermionic-Hirota-BKP}
\left[ h^{(1,\dots,p)}\otimes h^{(1,\dots,p)}, \sum_{a=1}^p\sum_{i\in\mathbb{Z}} \psi^{(a)}_i\otimes \psi^{\dag(a)}_i +  
\sum_{a=1,2}\sum_{i\in\mathbb{Z}} \psi^{\dag(a)}_i\otimes \psi^{(a)}_i +\varphi\otimes \varphi 
\right]=0
\ee
 From (\ref{fermionic-Hirota-BKP}) the multicomponent BKP Hirota equations are obtained \cite{KvdLbispec}:
\bea\label{Hirota-p-component-lBKP}
  \sum_{a=1}^p\oint\frac{dz}{2\pi i}z^{N^{(a)}{'}-N^{(a)}-2}e^{V(s^{(a)}{'}-s^{(a)},z)}
  \tau\left({\bf N}_-^{[a]}{'} ;{\bf s}_-^{[a]}{'}(z)\right)
  \tau\left({\bf N}_+^{[a]}  ; {\bf s}_+^{[a]}(z)\right)
   \nonumber\\
  + \sum_{a=1}^p \oint\frac{dz}{2\pi i}z^{N^{(a)}-N^{(a)}{'}-2}e^{V(s^{(a)}-s^{(a)}{'},z)}
  \tau\left({\bf N}_+^{[a]}{'} ;{\bf s}_+^{[a]}{'}(z)\right)
  \tau\left({\bf N}_-^{[a]} ;{\bf s}_-^{[a]}(z)\right)
   \nonumber\\
= \frac{1}{2}(1-(-1)^{\sum_{a=1}^p(N^{(a)}{'}+N^{(a)})})
   \tau\left({\bf N}{'} ; {\bf s}{'}\right)
  \tau\left({\bf N} ; {\bf s}\right) \quad
\eea
where 
\[
{\bf N}_\pm^{[a]} :=\left(N^{(1)},\dots,N^{(a-1)},N^{(a)}\pm 1,N^{(a+1)},\dots,N^{(p)} \right)
\]
\[
{\bf s}_\pm ^{[a]}(z):= \left(s^{(1)},\dots,s^{(a-1)}, s^{(a)}\pm [z^{-1}],s^{(a+1)},\dots, s^{(p)} \right)
\]
In (\ref{Hirota-p-component-lBKP}), ${\bf N}=\left(N^{(1)},\dots,N^{(p)} \right)$ and 
${\bf N}{'}=\left(N^{(1)}{'},\dots,N^{(p)}{'} \right)$ are two independent sets of vacuum charges, while
$s^{(a)}=\left( s^{(a)}_1,s^{(a)}_2, s^{(a)}_3, \right)$ and
$s^{(a)}{'}=\left( s^{(a)}_1{'},s^{(a)}_2{'}, s^{(a)}_3{'}, \right)$, $a=1,\dots,p$, are two independent sets 
of the multicomponent BKP higher times.

\paragraph{Pfaffian.} If $A$ an anti-symmetric matrix of an odd order its determinant
vanishes. For even order, say $k$, the following multilinear form
in $A_{ij},i<j\le k$
 \be\label{Pf''}
\Pf [A] :=\sum_\sigma
{\sgn(\sigma)}\,A_{\sigma(1),\sigma(2)}A_{\sigma(3),\sigma(4)}\cdots
A_{\sigma(k-1),\sigma(k)}
 \ee
where sum runs over all permutation restricted by
 \be
\sigma:\,\sigma(2i-1)<\sigma(2i),\quad\sigma(1)<\sigma(3)<\cdots<\sigma(k-1),
 \ee
 coincides with the square root of $\det A$ and is called the
 {\em Pfaffian} of $A$. As one can see the Pfaffian  contains
 $1\cdot  3\cdot 5\cdot \cdots \cdot(k-1)=:(k-1)!!$ terms.

 \paragraph{BKP tau functions} \cite{OST-I}. A class of BKP tau functions has the following form
 \[
 \tau^{\rm BKP}(N,n,\bpow;A)=\,\sum_{h_1>\cdots >h_N\ge 0}\,{\bar A}_h(n)\,s_{\{ h\}}(\bpow) 
 \]
where $s_{\{ h\}}:=s_\lambda$, $h_i=\lambda_i-i+N$, $i=1,\dots,N$. 
  The factors ${\bar A}_h(n)$ on the right-hand side are
   determined in terms a pair $(A, a)=:{\bar A}$ where $A$ is an infinite skew symmetric matrix and $a$
an infinite vector. For a strict partition $h=
(h_1,\dots,h_N )$, the numbers
${\bar A}_h(n)$ are defined as the Pfaffian of an antisymmetric $2k
\times 2k$ matrix ${\tilde A}$ as follows:
  \be
  \label{A-c}
{\bar A}_{h}(n):=\,\Pf[{\tilde A}]
  \ee
where for $N=2k$ even
  \be
  \label{A-alpha-even-n}
{\tilde A}_{ij}=-{\tilde A}_{ji}:=A_{h_i+n,h_j+n},\quad 1\le
i<j \le 2k
  \ee
and for $N=2k-1$ odd
 \be \label{A-alpha-odd-n} {\tilde
A}_{ij}(n)=-{\tilde A}_{ji}(n):=
\begin{cases}
A_{h_i+n,h_j+n} &\mbox{ if }\quad 1\le i<j \le 2k-1 \\
a_{h_i+n} &\mbox{ if }\quad 1\le i < j=2k .
 \end{cases}
  \ee
In addition we set ${\bar A}_0 =1$.

The fermionic realization for this tau function is
\[
\tau^{\rm BKP}(N,n,\bpow;A)\,=\,\l N+n| \,e^{\sum_{m>0}\,\frac{J_mp_m}{m}\,}\, e^{\sum_{i>j}\,A_{ij}\psi_i\psi_j +
\sqrt{2}\sum_i\,a_i\psi_i \phi}\,|n\r
\]
see \cite{OST-I}.

 \paragraph{2KP tau functions}. A class of 2KP tau functions has the following form
 \[
 \tau^{\rm 2KP}(N,n,\bpow^{(1)},\bpow^{(2)};B)=\,
 \sum_{h^{(1)}_1>\cdots >h^{(1)}_N\,\ge 0\atop h^{(2)}_1>\cdots >h^{(2)}_N\,\ge 0
 }\, B_{h^{(1)},h^{(2)}}(n)\,s_{\{ h^{(1)}\}}(\bpow^{(1)}) s_{\{ h^{(2)}\}}(\bpow^{(2)})
 \]  
  The factors $B_{h^{(1)},h^{(2)}}(n)$ on the right-hand side is
   determined in terms an infinite matrix $B$. For a pair of strict partitions $h^{(i)}$, $i=1,2$,
 the numbers $B_{h^{(1)},h^{(2)}}(n)$ are defined as follows:
  \be
  \label{B}
B_{h^{(1)},h^{(2)}}(n):=\,\det\left[B_{h^{(1)}_i+n,h^{(2)}_j+n}\right]_{i,j=1,\dots,N}
  \ee
compare to \cite{TakasakiSchur}.

The fermionic realization for this tau function is
\[
\tau^{\rm 2KP}(N,n,\bpow^{(1)},\bpow^{(2)};B)\,=\,
\l N+n,-N+n| \,e^{\sum_{i=1,2}\sum_{m>0}\,\frac{J^{(i)}_mp^{(i)}_m}{m}\,}\, 
e^{\sum_{i,j}\,B_{ij}\psi^{(1)}_i\psi^{\dag(2)}_{-1-j}}\,|n,n\r
\]
For skew-symmetric $B$ this tau function is the square of the DKP tau function \cite{LeurO-2014}.

\paragraph{Determinants and pfaffians of special matrices.} 
We will look at determinants of degenerate matrices of form
\be\label{degenerateA}
A_{ij}=b_i c_j\,,\quad i,j=1,\dots,n
\ee
where $b_i$ and $c_j$ are odd Grassmannian numbers. We see that
\be
\det A = n! b_1\cdots b_n c_n\cdots c_1
\ee
For instance for a $2\times 2$ matrix (\ref{degenerateA}) we have $\det A=2b_1b_2c_2c_1$.

Now, let $A$ is a skew-symmetric matrix given by
\be\label{skew-quasi-gen}
A_{ij}=b_i c_j\,,i<j\,,\quad i,j=1,\dots,2n
\ee
Then for both $b_i$ and $c_i$ are odd Grassmannian numbers we get
\be\label{skew-quasi-gen-pf-odd}
\Pf A = \,(2n-1)!!\,b_1\cdots c_{2n}b_{2n}\cdots b_1
\ee
For both $b_i$ and $c_i$ are even Grassmannian numbers we obtain
\be\label{skew-quasi-gen-pf-even}
\Pf A = \,b_1\cdots b_{2n}c_{2n}\cdots c_1
\ee

\paragraph{Exponentials.}
Let $\xi_i$ and $\eta_i$ are odd Grassmannian numbers.
We have
\be\label{exp=sum-det}
e^{\sum_{i,j} \,A_{ij}\xi_i\eta_j}=1+\sum_{k>0}A_{(\alpha_1,\dots,\alpha_k|\beta_1,\dots,\beta_k)}
\xi_{\alpha_1}\cdots\xi_{\alpha_1} \eta_{\beta_k}\cdots\eta_{\beta_1}
\ee
where
\be\label{det-A}
A_{(\alpha_1,\dots,\alpha_k|\beta_1,\dots,\beta_k)}=\det\left( A_{\alpha_i\beta_j} \right)_{i,j=1,\dots,k}
\ee
Now let $A$ is a skew-symmetric matrix (\ref{skew-quasi-gen}) and $\{a_i\}\,,i>0$ is a set of (even) numbers.

\paragraph{Quasi-tau functions}

One can consider the following series in the Schur functions (compare to \cite{AMMN-2014})
\be\label{quasitau}
\tau^{[\textsc{e}]}(n,\{\bpow^{i}\} ):= g(n)\sum_{\lambda\in\Pa} r_\lambda(n)\prod_{i=1}^\textsc{e} s_\lambda(\bpow^{(i)})
\ee
It may be presented in forms
\be
\tau^{[\textsc{e}]}(n,\{\bpow^{i}\} )= g(n)\l n+N|e^{\sum_{m>0}^k p^{(1)}_m B_m+
\sum_{i=2}^k p^{(i)}_m J^{(i)}_m} e^{\omega^{[\textsc{e}]}} |0\r =
\l n+N|e^{
\sum_{i=1}^k p^{(i)}_m J^{(i)}_m} e^{-\sum_{i\in\mathbb{Z}}U_i :\psi_i\psi^\dag_i:} e^{\omega^{[\textsc{e}]}}  |0\r
\ee
where
\be\label{omega-e}
\omega^{[\textsc{e}]}:=\,\sum_{i>j} \prod_{a=1}^\textsc{e} \left(\psi_i^{(a)}\psi_j^{(a)}\right)\,-
\,\sqrt{2}\,\phi\sum_{i\in\mathbb{Z}} \prod_{a=1}^\textsc{e}\psi_i^{(a)}
\ee
which may be viewed as a sum of commutative elements in the tensor product
${\hat o}(2\infty+1)\otimes \cdots \otimes {\hat o}(2\infty+1)$.

Let us note that
\be\label{omega-generating-tensor-Fock}
 e^{\omega^{[\textsc{e}]}} |n\r \,=\,
 1+\sum_{N>0} \, c_\textsc{e}\,\underbrace{|\lambda,N+n \r \otimes\cdots \otimes|\lambda,N+n}_\textsc{e}\r
\ee
where $c_\textsc{e}=1$ for even $\textsc{e}$ and $c_\textsc{e}=(2\textsc{e}-1)!!$ for $\textsc{e}$ odd.

 By coupling (\ref{omega-generating-tensor-Fock}) with the vector
\[
\l N+n | \, e^{\sum_{a=1}^\textsc{e}\sum_{m>0} \frac 1m J^{(a)}_m p^{(a)}_m}\, e^{-\sum_{i\in\mathbb{Z}} U_i: \psi_i\psi_i^\dag :}
\]
we obtain the right hand side of (\ref{quasitau}).

For $\textsc{e}=1$ we obtain hypergeometric $\tau_r^{\rm BKP}(N,n,\bpow)$. For $\textsc{e}=2$ we obtain
$\tau_r^{\rm TL}(N,n,\bpow^{(1)},\bpow^{(2)})$.

\section{Matrix integrals as generating functions of Hurwitz numbers and of Hurwitz generating series
\label{Matrix-integrals-appendix}}

The task of this section is to present Hurwitz generating series (\ref{hyp-tau-e-f}) in form of matrix integrals. 
The idea is that various integrals of products of the 2KP and BKP tau functions of matrix argument
\footnote{
If $\bpow=(p_1,p_2,\dots)$ and $p_m=\tr M^m$ where $M$ is a matrix, we may write $\tau=\tau(M):=\tau\left(\bpow(M) \right)$ 
and call it
tau function of matrix argument} result in the series (\ref{hyp-tau-e-f}). 

Below we will write down matrix integrals which generate Hurwitz numbers themselves rather than weighted sums
(as (\ref{MM-Schur-series}) does),  though we also present some generating series for the sums.

\paragraph{Notation. Useful relations.}

We shall exploit the known formulae for integrations of Schur functions over the
unitary group and over complex matrices. Earlier they were used in \cite{O-2003} to clarify some links between
matrix models and integrable hierarchies. 

As in \cite{O-2003} instead of power sums written by small bold characters (like $\bpow$) sometimes where it is suitable 
we shall use the matrix arguments written by large character. Say, $s_\lambda(X)$ and $\tau(X)$ are respectively equal 
to  $s_\lambda(X):=s_\lambda\left(\bpow(X)\right)$ and  $\tau(X):=\tau \left(\bpow(X)\right)$ where
 $\bpow(X)=(p_1(X),p_2(X),\dots)$, where $p_m(X)=\tr X^m$.

 We use the following notations
 \begin{itemize}
  \item $  d_*U $ is the normalized Haar measure on $\mathbb{\mathbb{U}}(N)$: $\int_{\mathbb{U}(N)}d_*U =1$
  
  \item $Z$ is a complex matrix, $Z=UX(1+J)U^\dag$ (the Schur decomposition), where $X=\diag (z_i)$ is diagonal, 
  $J$ is strictly upper triangle, $U\in\mathbb{\mathbb{U}}(N)$
    $$
d\Omega^\texttt{C}(Z,Z^\dag)  =\,\pi^{-n^2}\,e^{-\tr \left(ZZ^\dag\right)}\,
\prod_{i,j=1}^N \,d \Re Z_{ij}d \Im Z_{ij} 
  $$
  \[
 = \,c_\texttt{Z} \,
d_*U\,||\,\prod_{N\ge i>j}\,|z_i-z_j|^2\,
\prod_{i=1}^N\,e^{-|z_i|^2}\,d^2z_i \,\left[e^{-\tr JJ^\dag} d^2 J_{ij}\right]  
  \]
  where the part related to the upper triangular factor in brackets is not important for our problems.
  
  \item $M$ is a normal matrix, $Z=UXU^\dag$, where $X=\diag (z_i)$ is diagonal,  
  $U\in\mathbb{\mathbb{U}}(N)$  
  $$  
  d\Omega^\texttt{N}(M,M^\dag)\,= \,\pi^{-n^2}\,e^{-\tr\left(MM^\dag \right)}\,
\prod_{i,j=1}^N d \Re M_{ij}d \Im M_{ij}  
  $$
  \[
   =\,c_\texttt{M}\, d_*U \,
\prod_{N\ge i>j}|z_i-z_j|^2\,\prod_{i=1}^N \,e^{-|z_i|^2}
\,d^2 z_i
  \]

  \item $H^{(1)}$ is a Hermitian matrix and $H^{(2)}$ is anti-Hermitian one, 
  $H^{(c)}=U^{(c)}X^{(c)}U^{(c)\dag}$, $X^{(c)}=\diag(x_i^{(c)})$,
  $U,U^{(c)}\in\mathbb{U}(N)$, $c=1,2$. Measures 

 $$
 d\Omega^\texttt{HU}(H^{(1)},H^{(2)},U)=
e^{-\tr \left(H^{(1)} U H^{(2)} U^\dag \right)} d_*U\,\prod_{i\le j} 
d\Re H^{(1)} d\Im H^{(2)}\prod_{i<j} d\Im H^{(1)} d\Re H^{(2)} ,
  $$
  
  $$
   d\Omega^\texttt{H}(H^{(1)},H^{(2)})= \,\int_{\mathbb{U}(N)}\,
e^{-\tr \left(H^{(1)} UH^{(2)}U^\dag\right)}d_*U\, \prod_{i\le j} 
d\Re H^{(1)} d\Im H^{(2)}\prod_{i<j} d\Im H^{(1)} d\Re H^{(2)}
  $$ 
  \[
  =\,c_\texttt{H}\, \prod_{c=1,2} d_*U^{(c)}
\prod_{N\ge i>j}(x^{(c)}_i-x^{(c)}_j)\prod_{i=1}^N e^{-x^{(1)}_i x^{(2)}_i}d x^{(1)}_i d  x^{(2)}_i
  \]
 \end{itemize}
where the constants $c_a$, $a=\texttt{C,N,H}$,  are chosen for normalization: $\int d\Omega_\rho^{(a)}=1$.
\br \label{notation-M-M^*-1} In what follows,
for unification and to save space,
 we shall use the notation $M$ and $M^*$ replacing the pairs $Z,Z^\dag$, $M,M^\dag$ and also $H^{(1)},H^{(2)}$.
In the last case the matrices $M$ and $M^*$ are not related by the Hermitian conjugation.
\er

These measures provides the relation
\be\label{s-s-N_lambda}
\int s_\lambda(M)s_\mu(M^*)\,d\Omega^a(M,M^*) = (N)_\lambda\delta_{\lambda,\mu}
\ee
where $a=\texttt{C},\texttt{N},\texttt{H}$. This relation was used in 
\cite{O-Acta},\cite{HO-2MM},\cite{O-2003},\cite{AMMN-2014},\cite{OShiota-2004}, for for models of Hermitian, complex 
and normal matrices.
\footnote{If we 
replace the factor $e^{\tr\left(MM^* \right)}$ in the measure 
$d\Omega^a$ by a hypergeometric tau function $\tau_r(N,MM^*,I_N)$, then the factor $(N)_\lambda$ in the right 
hand side of (\ref{s-s-N_lambda}) should be replaced by $\frac{1}{r_\lambda(N)}$ \cite{O-Acta}.}

By $I_N$ we shall denote the $N\times N$ unit matrix.
Then (for instance see \cite{Mac})
\bl \label{useful-relations}
Let $A$ and $B$ be normal  matrices (i.e. matrices diagonalizable by unitary transformations). Then
\begin{equation}\label{sAUBU^+}
\int_{\mathbb{U}(N)}s_\lambda(AUBU^{-1})d_*U=
\frac{s_\lambda(A)s_\lambda(B)}{s_\lambda(I_n)} \ ,
\end{equation}
For $A,B\in GL(N)$ we have
\begin{equation}\label{sAUU^+B}
\int_{\mathbb{U}(n)}s_\mu(AU)s_\lambda(U^{-1}B)d_*U=
\frac{s_\lambda(AB)}{s_\lambda(I_N)}\delta_{\mu,\lambda}\,.
\end{equation}
Below ${\bf p}_{\infty}=(1,0,0,\dots)$. 
\begin{equation}\label{sAZBZ^+}
\int_{\mathbb{C}^{n^2}} s_\lambda(AZBZ^+)e^{-\textrm{Tr}
ZZ^+}\prod_{i,j=1}^n d^2Z=
\frac{s_\lambda(A)s_\lambda(B)}{s_\lambda({\bf p}_{\infty})}
\end{equation}
and
\begin{equation}\label{sAZZ^+B}
\int_{\mathbb{C}^{n^2}} s_\mu(AZ)s_\lambda(Z^+B) e^{-\textrm{Tr}
ZZ^+}\prod_{i,j=1}^nd^2Z= \frac{s_\lambda(AB)}{s_\lambda({\bf
p}_{\infty})}\delta_{\mu,\lambda}\,.
\end{equation}
\el
We recall that
$$ s_\lambda(I_N)=(N)_\lambda s_\lambda(\bpow_\infty)\,,\qquad\chi_\lambda(1)=|\lambda|!s_\lambda(\bpow_\infty) $$
\br
Usually the relation (\ref{sAZBZ^+}) is written down for positive matrices $A$ and $B$. Equation 
(\ref{sAZBZ^+}) may be derived using the Gauss integration. Let us note that for any $A,B\in GL(N)$
the Gauss integrals
of products of type $\prod_{i} \left(\tr C^{k_i}\right)^{m_i}$ where $C=AZBZ^\dag$ yields sums of terms 
$\prod_{i} \left(\tr A^{k'_i}\right)^{m'_i}\prod_{i} \left(\tr A^{k''_i}\right)^{m''_i}$ which depend only on the spectrums
of matrices $A$ and $B$. 
\er

Lemma \ref{useful-relations} allows to pick up Hurwitz numbers from matrix integrals in many ways. Below we describe 
a set of the most natural ones.

First of all, step by step applying (\ref{sAZBZ^+}) we arrive at

\be\label{sAZBZ^+-corollary}
\int_{\left[C^{N^2}\right]^{\times(\textsc{f}-1)}}
s_\lambda\left(A_{\textsc{f}}\left(Z_{\textsc{f}-1}A_{\textsc{f}-1}Z_{\textsc{f}-1}^{\dag} \cdots Z_1A_1Z_1^\dag\right)\right)
\prod_{i=1}^{\textsc{f}-1} d\Omega^{\texttt{C}}(Z_i,Z_i^\dag)=
\frac{\prod_{i=1}^{\textsc{f}}  s_\lambda(A_i)}{\left(s_\lambda(\bpow_\infty)\right)^{\textsc{f}-1}}
\ee

We obtain the multi-matrix analogues of the Itsykson-Zuber integral

\be\label{multi-Itsykson-Zuber-Z}
\int e^{\tr V\left( X_{\textsc{f}-1} M_1X_1M_1^* \cdots 
M_{\textsc{f}-2}X_{\textsc{f}-2}M_{\textsc{f}-2}^*,\,\bpow \right)} 
\prod_{i=1}^{\textsc{f}-1} d\Omega^\texttt{C}(M_i,M_i^*)\,=
\,\left(s_\lambda(\bpow_\infty)\right)^2\,
\frac{s_\lambda(\bpow)}{s_\lambda(\bpow_\infty)}
\prod_{i=1}^{\textsc{f}-1}\,
\frac{s_\lambda(X_i)}{s_\lambda(\bpow_\infty)}
\ee

The relation (\ref{multi-Itsykson-Zuber-Z}) gives the generating function for Hurwitz numbers on the sphere
with $\textsc{f}$ ramification points (compare to \cite{ChekhovAmbjorn}).

Actually we use the simplest (the so-called vacuum) 2-KP tau function 
\be\label{vac-tau-2KP}
\tau_1^{\rm 2KP}(X,\bpow)\,:=\,\sum_\lambda \,s_\lambda(X)\,s_\lambda(\bpow)=e^{\tr V(X,\bpow)},\quad
\tau_1^{\rm 2KP}(X,\bpow_\infty)\,=e^{\tr X}
\ee
where
\[
 V(z,\bpow)=\sum_{m>0}\,
 \frac{z^mp_m}{m}
\]
($z$ may be a number and may be a matrix) as the integrand in 
(\ref{multi-Itsykson-Zuber-Z}).
To get analogues of  (\ref{multi-Itsykson-Zuber-Z}) for the projective plane we use the simplest BKP tau function 
\cite{OST-I}
\be\label{vac-tau-BKP}
 \tau_1^{\rm BKP}(X)\,:=\,\sum_\lambda \,s_\lambda(X)\,=\,\prod_{N>i>j}\,(1-x_ix_j)^{-1}\,\prod_{i=1}^N\,(1-x_i)^{-1}
\ee
Then we have
\be\label{multi-Itsykson-Zuber-Z-E=1}
\int_{{C^{N^2}}^{\times(\textsc{f}-1)}} \tau_1^{\rm BKP}\left( X_{\textsc{f}} Z_1X_1Z_1^\dag \cdots 
Z_{\textsc{f}-1}X_{\textsc{f}-1}Z_{\textsc{f}-1}^\dag \right)  
\prod_{i=1}^{\textsc{f}-1}\,e^{-\tr Z_iZ_i^\dag} d^2Z_i\,=
\, s_\lambda(\bpow_\infty) \prod_{i=1}^{\textsc{f}}\,
\frac{s_\lambda(X_i)}{s_\lambda(\bpow_\infty)}
\ee
The last formula is the generating function for the Hurwitz numbers for projective plane with $\textsc{f}$ 
ramification points.

{\bf Example}. $\mathbb{CP}^1$ with three ramification point. In particular we have the following integrals
generating Hurwitz numbers
\[  
\sum_{\lambda\atop\ell(\lambda)\le N} \frac{s_\lambda(\bpow^{(1)})
s_\lambda(\bpow^{(2)})s_\lambda(\bpow^{(3)})}{s_\lambda(\bpow_\infty)}\, =
\]
\[
= \,\int e^{\tr V(M_1 M_2,\,\bpow^{(3)})}\prod_{i=1,2} e^{\tr V(M^*_i,\,\bpow^{(i)})} d\Omega^a(M_i,M^*_i)
\]
\[
=\,\int e^{\tr V(\Lambda M,\,\bpow^{(1)})+ \tr V(M^*,\,\bpow^{(2)})} \,d\Omega^a(M,M^*)\,,\quad p_m^{(3)}=\tr \Lambda^m
\]
\[
 =\,\int e^{\tr V\left( X_2 M_1X_1M_1^*,\,\bpow \right)} 
\prod_{i=1,2} d\Omega^\texttt{C}(M_i,M_i^*)\,,\qquad p_m^{(i)}=\tr \left(X_i \right)^m,\, i=1,2
\]
\[
 =\,\int e\, ^{\tr \left( X_3 M_1X_1M_1^* M_2X_2M_2^*\right)} \,
\prod_{i=1}^{3} d\Omega^\texttt{C}(M_i,M_i^*)\,,\qquad p_m^{(i)}=\tr \left(X_i \right)^m,\, i=1,2,3
\]
where $a=\texttt{C,N,H}$.

{\bf Example}. Coverings of $\mathbb{RP}^2$ with three ramification point. To get a generating integral of the related Hurwitz
numbers one can take a $\mathbb{CP}^1$ Hurwitz
integral with 4 ramification points and replace any of the 2KP tau functions under the integral by a BKP tau function.
For instance, 
\[
 \frac{\prod_{i=1}^4\,s_\lambda(\bpow^{(i)}) }{\left( s_\lambda(\bpow_\infty) \right)^{2}}=
\]
\[
 \int e^{\tr V(\Lambda M_1^*M_2^*,\,\bpow^{(3)})}\,\prod_{i=1,2}\,e^{\tr V(M_i,\,\bpow^{(i)})}\,d\Omega^{\texttt{C}}(M_i,M_i^*),
 \quad p^{(4)}_m=\tr \Lambda^m
\]
Then the generating integral for $\mathbb{RP}^2$ Hurwitz numbers for covering given by 3 ramification points may be 
written as
\[
 \frac{\prod_{i=1}^3\,s_\lambda(\bpow^{(i)}) }{\left( s_\lambda(\bpow_\infty) \right)^{2}}=
\]
\[
 =\,\int \,\tau_1^{\rm BKP}(\Lambda M_1^*M_2^*)\,\prod_{i=1,2}\,e^{\tr V(M_i,\,\bpow^{(i)})}\,d\Omega^{\texttt{C}}(M_i,M_i^*)
 \,,\qquad p^{(3)}_m\,=\,\tr\,\Lambda^m
\]
\[
 = \int e^{\tr V(\Lambda M_1^*M_2^*,\,\bpow^{(2)})}\,e^{\tr V(M_1,\,\bpow^{(1)})}\,\tau_1^{\rm BKP}(M_2)\,
 \prod_{i=1,2}\,d\Omega^{\texttt{C}}(M_i,M_i^*)\,,\qquad p^{(3)}_m\,=\,\tr\,\Lambda^m
\]

{\bf Example}. An analog of (\ref{MM-Schur-series}) for $\mathbb{RP}^2$ with three ramfication points with two 
arbitrary profiles at $0$ and at $\infty$ with fixed length in the third point:
\[
 \sum_{\lambda}\frac{s_\lambda(I_N)s_\lambda(\bpow^{(1)})s_\lambda(\bpow^{(2)})}{\left(s_\lambda(\bpow_\infty) \right)^2}
\]
\[
 = \,\int \,\tau^{\rm BKP}_1\left( M_1 M_2 \right)  \,\prod_{i=1,2} \,
  e^{V(\tr M^*_i,\,\bpow^{(i)})}\,d\Omega^a(M_i,M^*_i)\,
\]
\[
=\, \int \, e^{\tr\left(\Lambda M_1 M_2 \right) }\,\tau^{\rm BKP}_1(M_1^*)\,e^{\tr V(M^*_2,\,\bpow)} \,
\prod_{i=1,2}\,d\Omega^{\texttt{C}}(M_i,M^*_i)\,, \qquad p^{(2)}_m =\tr \Lambda^m
\]
\[
  =\,\int \, e^{\tr\left( M_1 M_2 M_3\right) }\,\tau^{\rm BKP}_1(M_3^*) \,
  e^{\tr V(M^*_1,\,\bpow^{(1)})+\tr V(M^*_2,\,\bpow^{(2)})}\,
  \prod_{i=1}^3 \,d\Omega^{\texttt{C}}(M_i,M^*_i)
\]  
\[
  =\,\int\,  e^{\tr V\left( M_1 M_2 M_3,\,\bpow^{(1)}\right) }\,\tau^{\rm BKP}_1(M_3^*) \,e^{\tr (M^*_1+\tr V(M^*_2,\,\bpow^{(2)})}
  \,  \prod_{i=1}^3 \,d\Omega^{\texttt{C}}(M_i,M^*_i)
\] 
where $a=\texttt{C,N,H}$.

We also need the following relation which follows applying (\ref{sAZBZ^+}) and the applying (\ref{sAUU^+B}):

\be
\int_{[\mathbb{U}(N)]^{\times\textsc{g}}}\int_{[\mathbb{C}^{N^2}]^{\times\textsc{g}}} 
s_\lambda(Y_\textsc{g})\prod_{i=1}^\textsc{g} d_*U_i \prod_{i=1}^{2\textsc{g}}d\Omega^\texttt{C}(M_i,M_i^*)=
\, \left( s_\lambda(\bpow_\infty)\right)^{-2\textsc{g}}
\ee
where
\[
Y_\textsc{g}\,=\,\left(M_{2\textsc{g}} U_\textsc{g} M_{2\textsc{g}}^*M_{2\textsc{g}-1} U_\textsc{g}^\dag M_{2\textsc{g}-1}^*  
\right)\cdots 
\left(M_2 U_1 M_2^* M_1 U_1^\dag M_1^*  \right)
\]

\be
\int s_\lambda(\Lambda M_1^*\cdots M_{\textsc{g}}^*)\,
\prod_{i=1}^{\textsc{g}}\,\tau^{\rm BKP}_1(M_i)\,d\Omega^{\texttt{C}}(M_i,M_i^*)=
\frac{s_\lambda(\Lambda)}{s_\lambda(\bpow_\infty)^{\textsc{g}}}
\ee

We get

Let $M_1,\cdots , M_{\textsc{g}}$,   $\textsc{g}>0 $, be $N\times N$ normal matrices, that is they may be presented as 
$M_i=U_iX_iU_i^{-1}$
where where $U_i\in\mathbb{U}(N)$ and $(X_i)_{ab}=x^{(i)}_a\delta_{a,b}$, with $ x^{(i)}_a\in\mathbb{C},\,a=1,\dots,N$, being 
the eigenvalues of $M_i$. Let $\Lambda$ be diagonal.

\be
\tau_r^{\textsc{e},\textsc{f}+\textsc{g}}
\left(N,n,\Lambda,\bbpow^{(1)},\dots ,\bbpow^{(\textsc{f}-1)},\bpow^{(1)},\dots ,\bpow^{(\textsc{g})}\right):=
\ee
\be
\sum_{\lambda\in\Pa \atop \ell(\lambda)\le N} \,r_\lambda(n)
\left(s_\lambda(\bpow_\infty)\right)^{\textsc{e}}\frac{s_\lambda(\Lambda)}{s_\lambda(\bpow_\infty)}\,
\prod_{i=1}^{\textsc{f}-1}
\frac{s_\lambda(\bbpow^{(i)})}{s_\lambda(\bpow_\infty)}\,\prod_{i=1}^{\textsc{g}}\,
\frac{s_\lambda(\bpow^{(i)})}{s_\lambda(\bpow_\infty)} =
 \ee
 \be
 \label{int-Hurwitz-fixed-E-r-general}
\frac 1c \int  \,
\tau_{\tilde{r}}^{\textsc{e},\textsc{f}}
\left(N,n,\Lambda M^*_1\cdots M^*_{\textsc{g}},\bbpow^{(1)},\dots ,\bbpow^{(\textsc{f}-1)}\right)
\prod_{i=1}^{\textsc{g}} \,e^{\tr\, V\left(M_i,\bpow^{(i)}\right)}\, d\Omega^\alpha_\rho(M_i,M_i^*)
\ee
where $\alpha=\texttt{N,C,H}$ i.e. each $M_i$ and is respectively normal, complex matrix 

where $c=\int d\mu_\rho(M,M^\dag)$,
and where $\rho$ is rather arbitrary which provides $c$ to be finite, and which defines the relation between functions 
$r$ and  $\tilde{r}$ as follows: 	 	
\be\label{tilde-r-r}
n! \,{\tilde r}(1)\cdots {\tilde r}(n)=r(1)\cdots r(n)\,\int_{z\ge 0}\,z^n e^{-z}\,\rho(z) dz
\ee
In particular, for $\rho=1$, ${\tilde r}=r$.

Proof follows from the following relations:
\[
 e^{\tr V\left(M_i,\bpow^{(i)}\right) }=
 \sum_\lambda s_\lambda\left( M_i\right) s_\lambda\left(\bpow^{(i)}\right)\,,
\]
\[
\tau_r^{\textsc{e},\textsc{f}}
\left(N,n,\Lambda M^\dag_1\cdots M^\dag_{\textsc{g}},\bbpow^{(1)},\dots ,\bbpow^{(\textsc{f}-1)}\right)=
\]
\[
\sum_{\lambda\in\Pa \atop \ell(\lambda)\le N} r_\lambda(n)
\left(s_\lambda(\bpow_\infty)\right)^{\textsc{2}}
\frac{s_\lambda(\Lambda M^\dag_1 \cdots M^\dag_{\textsc{f}-1})}{s_\lambda(\bpow_\infty)}
\prod_{i=1}^{\textsc{f}-1}
\frac{s_\lambda(\bbpow^{(i)})}{s_\lambda(\bpow_\infty)}
\]
 then, having in mind that 
 $s_\lambda\left(UMU^\dag\right)=s_\lambda(M)$,  from (\ref{tilde-r-r}) which results in
\be\label{s-s-N-lambda-scalar-product}
 r_\lambda(N)\int s_\lambda(M)s_\mu(M^\dag)  d\mu(M,M^\dag) =
 \,{\tilde r}_\lambda(N)\,(N)_\lambda\,\delta_{\lambda,\mu}\,,
\ee
used in \cite{OShiota-2004}\footnote{see \cite{HO-scalar-prod} where this was used to define a deformed scalar product.}.
 
 For $\tau_r^{\textsc{e},\textsc{f}}$ we take the simplest 
 (the so-called, vacuum)  TL tau function, we obtain

The generating function for the Hurwitz numbers on the sphere with $\textsc{f}$ ramification points may be written as follows:
\be
\tau^{2,\textsc{f}}
\left(N,\Lambda,\bpow^{(1)},\dots ,\bpow^{(\textsc{f}-1)}\right):=\,
\sum_{\lambda\in\Pa \atop \ell(\lambda)\le N} \left(s_\lambda(\bpow_\infty)\right)^{2}
\frac{s_\lambda(M)}{s_\lambda(\bpow_\infty)}\prod_{i=1}^{\textsc{f}-1}
\frac{s_\lambda(\bpow^{(i)})}{s_\lambda(\bpow_\infty)}
\ee
\be
\label{int-Hurwitz-sphere-general}
=\,\int   e^{\tr\,V\left( \Lambda M^*_1\cdots M^*_{\textsc{f}-2},\bpow^{(\textsc{f}-1)} \right)}\prod_{i=1}^{\textsc{f}-2}
e^{\tr\, V\left(M_i,\bpow^{(i)}\right)} d\Omega^a \left(M_i,M_i^*\right)
\ee
For independent proof one may use
\[
 e^{\tr V\left( \Lambda M^\dag_1\cdots M^*_{\textsc{f}-2},\bpow^{\textsc{f}-1} \right)}=
 \sum_\lambda s_\lambda\left( \Lambda M^*_1 \cdots M^\dag_{\textsc{f}-1} \right) s_\lambda\left(\bpow^{\textsc{f}-1}\right)\,,
\]

The generating function for the Hurwitz numbers on the projective plane with $\textsc{f}$ ramification points may 
be written as follows:
\[
\tau^{1,\textsc{f}}
\left(N,\Lambda,\bpow^{(1)},\dots ,\bpow^{(\textsc{f}-1)}\right):=\,
\sum_{\lambda\in\Pa \atop \ell(\lambda)\le N} \,s_\lambda(\bpow_\infty)\,
\frac{s_\lambda(\Lambda)}{s_\lambda(\bpow_\infty)}\,\prod_{i=1}^{\textsc{f}-1}\,
\frac{s_\lambda(\bpow^{(i)})}{s_\lambda(\bpow_\infty)} =
\]
\be
\label{int-Hurwitz-projective-plane-general}
\int 
\tau_1^{\rm BKP}\left(X\right)
\prod_{i=1}^{\textsc{f}-1}
e^{\tr\, V\left(M_i,\bpow^{(i)}\right)} d\mu\left(M_i,M_i^\dag\right) \,=
\ee
\be
\int 
\tau_1^{\rm 2KP}\left(X ,\bpow^{(\textsc{f}-1)}\right)
\tau_1^{\rm BKP}\left( M_{\textsc{f}-2} \right)
  d\mu\left(M_{\textsc{f}-1},M_{\textsc{f}-1}^\dag\right)  \prod_{i=1}^{\textsc{f}-2}
e^{\tr\, V\left(M_i,\bpow^{(i)}\right)} d\mu\left(M_i,M_i^\dag\right)
\ee
where
\[
 X=\Lambda M^\dag_1\cdots M^\dag_{\textsc{f}-2}
\]
and
\[
 \tau_1^{\rm BKP}(X)\,:=\,\sum_\lambda \,s_\lambda(X)\,=\,\prod_{N>i>j}\,(1-x_ix_j)^{-1}\,\prod_{i=1}^N\,(1-x_i)^{-1}
\]

Obtained relations allows 
to change $\textsc{f}$ keeping the Euler characteristic $\textsc{e}$
Now, we want to change  $\textsc{e}$. This is even easier and may be done in a straightforward way

\be
\tau^{\textsc{e}-2,\textsc{f}-2}
\left(N,\bpow^{(1)},\dots ,\bpow^{(\textsc{f}-2)}\right)=\int_{\mathbb{U}(N)}\,
\tau^{\textsc{e},\textsc{f}}
\left(N,\bpow^{(1)},\dots ,\bpow^{(\textsc{f}-2)},U,U^\dag\right)\,d_*U
\ee

Now, we can construct generating functions
for $H^\textsc{e}_{d,N}$ where for even $\textsc{e}$ we start from a Toda lattice tau function and
for $\textsc{e}$ odd we start from a BKP tau function.

\[
 d\Omega(Z,Z^\dag)=\,e^{-\tr Z_iZ_i^\dag} d^2Z_i
\]

\[
 \,\sum_{\lambda\atop \ell(\lambda)\le N}\,\left(s_\lambda(\bpow_\infty)\right)^{2-2\textsc{g}}
 \,\frac{s_\lambda(\bpow)}{s_\lambda(\bpow_\infty)}\,
 \prod_{i=1}^{\textsc{f}}\,\frac{s_\lambda(\Lambda_i)}{s_\lambda(\bpow_\infty)}=
\]
\be
\int \, \tau_1^{\rm 2KP}\left( 
X_{\textsc{f}} Y_{\textsc{g}},\bpow \right) \, \prod_{i=1}^{\textsc{f}}\,d\Omega(Z_i,Z_i^\dag)
\prod_{i=1}^{\textsc{g}}\, d_* U_i\,\prod_{i=1}^{2\textsc{g}}  d\Omega(C_i,C_i^\dag)\,=
\ee
\be
\int \, \tau_1^{\rm 2KP}\left( 
{\tilde X}_{\textsc{f}} {\tilde Y}_{\textsc{g}},\bpow \right) \, \prod_{i=1}^{\textsc{f}}\, d_*U_i
\prod_{i=1}^{\textsc{g}}\, d\Omega(Z
_i,Z_i^\dag)\,\prod_{i=1}^{2\textsc{g}} \prod_{i=1}^{\textsc{f}}\,d_*W_iDDXSSWSW
\ee

 and
\be
\int \, \tau_1^{\rm BKP}\left( 
X_{\textsc{f}} Y_{\textsc{g}} \right) \, \prod_{i=1}^{\textsc{f}}\,e^{-\tr Z_iZ_i^\dag} d^2Z_i
\prod_{i=1}^{\textsc{g}}\, d_* U_i\,\prod_{i=1}^{2\textsc{g}} e^{-\tr C_{i}C_{i}^\dag} d^2C_{i}\,=
\ee
\[
 =\,\sum_{\lambda\atop \ell(\lambda)\le N}\,\left(s_\lambda(\bpow_\infty)\right)^{1-2\textsc{g}}
 \prod_{i=1}^{\textsc{f}}\,\frac{s_\lambda(\Lambda_i)}{s_\lambda(\bpow_\infty)}
\]
where
\[
 X_\textsc{f} =\Lambda_{\textsc{f}}\left(Z_{\textsc{f}-1}\Lambda_{\textsc{f}-1} Z_{\textsc{f}-1}^\dag \cdots 
Z_{1}\Lambda_{1} Z_{1}^\dag \right)
\]
\[
Y_\textsc{g}=\left(C_{2\textsc{g}} U_{\textsc{g}} C_{2\textsc{g}}^\dag 
C_{2\textsc{g}-1} U_{\textsc{g}}^\dag C_{2\textsc{g}-1}^\dag \right) \cdots 
\left(C_{2} U_1 C_{2}^\dag C_{1} U_1^\dag C_{1}^\dag \right)
\]
and $Z_i$ and $C_i$ are complex matrices.

\be
\int \, \tau_1^{\rm 2KP}\left( 
X_{\textsc{f}} Y_{\textsc{g}},\bpow \right) \, \prod_{i=1}^{\textsc{f}}\,e^{-\tr Z_iZ_i^\dag} d^2Z_i
\prod_{i=1}^{\textsc{g}}\, d_* U_i\,\prod_{i=1}^{2\textsc{g}} e^{-\tr C_{i}C_{i}^\dag} d^2C_{i}\,=
\ee
\[
 =\,\sum_{\lambda\atop \ell(\lambda)\le N}\,\left(s_\lambda(\bpow_\infty)\right)^{2-2\textsc{g}}
 \,\frac{s_\lambda(\bpow)}{s_\lambda(\bpow_\infty)}\,
 \prod_{i=1}^{\textsc{f}}\,\frac{s_\lambda(\Lambda_i)}{s_\lambda(\bpow_\infty)}
 \]
 and
\be
\int \, \tau_1^{\rm BKP}\left( 
X_{\textsc{f}} Y_{\textsc{g}} \right) \, \prod_{i=1}^{\textsc{f}}\,d\Omega^a(M_i,M_i^*)
\prod_{i=1}^{\textsc{g}}\, d_* U_i\,\prod_{i=1}^{2\textsc{g}} d\Omega^\texttt{C}(Z,Z^\dag)\,=
\ee
\[
 =\,\sum_{\lambda\atop \ell(\lambda)\le N}\,\left(s_\lambda(\bpow_\infty)\right)^{1-2\textsc{g}}
 \prod_{i=1}^{\textsc{f}}\,\frac{s_\lambda(\Lambda_i)}{s_\lambda(\bpow_\infty)}
\]
where
\[
 X_\textsc{f} =\Lambda_{\textsc{f}}\left(U_{\textsc{f}-1}M_{\textsc{f}-1} U_{\textsc{f}-1}^\dag \cdots 
U_{1}M_{1} U_{1}^\dag \right)\,,\quad
Y_\textsc{g}=\left(Z_{2\textsc{g}} U_{\textsc{g}} Z_{2\textsc{g}}^\dag 
Z_{2\textsc{g}-1} U_{\textsc{g}}^\dag Z_{2\textsc{g}-1}^\dag \right) \cdots 
\left(Z_{2} U_1 Z_{2}^\dag Z_{1} U_1^\dag Z_{1}^\dag \right)
\]
and $M_i,M_i^*$ may be conjugated complex matrices, conjugated normal ones, or a pair of two nonrelated
matrices, one being Hermitian matrix while the other is an anti-Hermitian one,
$U_i\in\mathbb{U}$ and $Z_i$ are complex matrices.

Then we use the known formula (for instance see \cite{Mac})
\be\label{int-Schurs}
 \int_{\mathbb{U}(N)} s_\lambda(AUBU^\dag) dU = \frac{ s_\lambda(A) s_\lambda(B)}{s_\lambda(I_N)}
\ee
where $I_N$ is the unit matrix (see for instance \cite{Mac}) and Cauchy-Littlewood formula
\be
 e^{\sum_{m>0}\frac 1m p_mp_m^*}=\sum_\lambda s_\lambda(\bpow)s_\lambda(\bpow^*)
\ee
where
\be\label{p^*_m}
  p^*_m =  \tr (AUBU^\dag)^m
\ee
like it was done in \cite{O-2002}.
Writing the product of $K$ matrices as $A_1 B_1$ where $B_1=A_2\cdots A_K$ diagonalizing, then repeating $K-1$ times
we obtain
\be
I={\rm Vol}{\mathbb{U}}(N)\sum_{\lambda \atop \ell(\lambda)\le N}
\frac{s_\lambda(\bpow)\prod_{i=1}^K s_\lambda(A_i)}{(s_\lambda(I_N))^{K-1}}
\ee
which may be related to more complicated Hurwitz numbers with $K+1$ arbitrary profiles, namely
to the sums $S_{\mathbb{CP}^1}(d|b|l_1,\dots,l_k|l_1^*,\dots,l^*_s|
\Delta^{(1)},\dots,\Delta^{(K+1)})$.

Similarly, we can integrate hypergeometric $\tau_r^{TL}(\bpow,A_1\cdots A_K)$ and $\tau_r^{BKP}A_1\cdots A_K)$
(instead of the simplest
TL tau function given by Itsykson-Zuber $e^{\tr AUBU^\dag}$.

We obtain

The generating function
$S_{\mathbb{CP}^1}(d|b|l_1,\dots,l_k|l_1^*,\dots,l^*_s|
\Delta^{(1)},\dots,\Delta^{(K+1)})$ is constructed as the following matrix integral
\be\label{integral-1}
\int_{\mathbb{U}(N)\times \cdots \times \mathbb{U}(N)}\,
\tau_r^{\rm TL}(n,\bpow,U_1A_1U_1^\dag U_2A_2U_2^\dag \cdots U_KA_K U_K^\dag)\,\prod_{i=1}^K \,dU_i=
\ee
\be
\sum\limits_B\,(qe^{\beta n})^d\frac{\beta^b}{b!}\,
\frac{\prod_{i=1}^k \,(a_i+n)^{l_i}}{\prod_{i=1}^s \,b_i^d\,(b_i+n)^{l^*_i} }\,
S_{\mathbb{CP}^1}^N(B)\, \prod_{i=1}^{K+1}\,{\bpow}^{(i)}_{\Delta^{(i)}}
\ee

where the sum is taken by all
$B=(d|l_1,\dots,l_k|l_1^*,\dots,l^*_s| \Delta^{(1)},\cdots,\Delta^{(K+1)} )$
and where $\bpow^{(i)} = \left( p_1^{(i)},p_2^{(i)},\dots \right)$ and
$$
 p_m^{(i)}=\tr A_i^m\,,\quad i=1,\dots, K,\qquad p_m^{(K+1)}=p_m
$$

 From (\ref{int-Schurs})-(\ref{p^*_m}) we obtain
 \be\label{integrals-Schurs-TL}
\int_{\mathbb{U}(N)\times \cdots \times \mathbb{U}(N)}
\tau_r^{\rm TL}(\bpow,U_1A_1U_1^\dag U_2A_2U_2^\dag \cdots U_KA_K U_K^\dag)\prod_{i=1}^K dU_i =
\sum_{\lambda\atop \ell(\lambda)\le N}r_\lambda(n)\frac{s_\lambda(\bpow)
\prod_{i=1}^K s_\lambda(A_i)}{\left(s_\lambda(I_N) \right)^{K-1}}
\ee
 Then we apply the same steps as the Theorem \ref{Th-BKP}.

Construct now  generating function for $S^N_{\mathbb{RP}^2}(d|l_1,\dots,l_k|l_1^*,\dots,l^*_s|
\Delta^{(1)},\dots,\Delta^{(K)})$. Put
\be\label{integrals-of-BKP-generating-Hurwitz-sums}
\int_{\mathbb{U}(N)\times \cdots \times \mathbb{U}(N)}\,
\tau_r^{\rm BKP}(n,U_1A_1U_1^\dag U_2A_2U_2^\dag \cdots U_KA_K U_K^\dag)\,\prod_{i=1}^K \,dU_i\, =
\ee
\be
\sum\limits_B \,(qe^{\beta n})^d\frac{\beta^b}{b!}  \,\frac{\prod_{i=1}^s
\,b_i^{-d}\,(b_i+n)^{l^*_i} }{\prod_{i=1}^k a_i^{-d}(a_i+n)^{l_i}}\,
S^N_{\mathbb{RP}^2}(B)\, \prod_{i=1}^K \,{\bpow}^{(i)}_{\Delta}
\ee
where the sum is taken by all $B=(d|l_1,\dots,l_k|l_1^*,\dots,l^*_s|\Delta^{(1)},\dots,\Delta^{(K)})$
and where $\bpow _i = \left( p_1^{(i)}, p_2^{(i)},\dots \right)$ and
\[
 p_m^{(i)}=\tr A_i^m\,,\quad i=1,\dots, K
\]

 From (\ref{int-Schurs})-(\ref{p^*_m}) we obtain
 \be\label{integrals-Schurs-BKP}
\int_{\mathbb{U}(N)\times \cdots \times \mathbb{U}(N)}
\tau_r^{\rm BKP}(U_1A_1U_1^\dag U_2A_2U_2^\dag \cdots U_KA_K U_K^\dag)\,\prod_{i=1}^K dU_i =
\sum_{\lambda\atop \ell(\lambda)\le N}r_\lambda(n)\frac{\prod_{i=1}^K s_\lambda(A_i)}{\left(s_\lambda(I_N) \right)^{K-1}}
\ee
 Then we apply the same steps as the Theorem \ref{Th-BKP}.

\end{document}